\documentclass[ba]{imsart}

\RequirePackage{amsthm,amsmath,amsfonts,amssymb}
\RequirePackage[authoryear]{natbib}
\RequirePackage[colorlinks,citecolor=blue,urlcolor=blue,backref=page,backref=page]{hyperref}
\RequirePackage{graphicx}
\usepackage{graphicx}%
\usepackage{multirow}%
\usepackage{amsmath,amssymb,amsfonts}%
\usepackage{amsthm}%
\usepackage{mathrsfs}%
\usepackage[title]{appendix}%
\usepackage{xcolor}%
\usepackage{textcomp}%
\usepackage{manyfoot}%
\usepackage{booktabs}%
\newcommand{\widewedge}[1]{\overset{\scalebox{6}[.4]{\(\wedge\)}}{#1}}
\usepackage{listings}%

\usepackage{natbib,graphicx,epsfig,pdfpages,longtable,url,amsmath,bbm,amssymb,amsthm, afterpage, stmaryrd}
\usepackage{dsfont, amsmath}
\usepackage{hyperref}
\usepackage{pdfpages}
\usepackage[ruled]{algorithm2e}
\pubyear{2026}
\arxiv{2607.23998}
\volume{TBA}
\issue{TBA}
\firstpage{1}
\lastpage{1}

\startlocaldefs
\theoremstyle{plain}

\newtheorem{theorem}{Theorem}[section]
\newtheorem{proposition}{Proposition}[section]

\theoremstyle{definition}

\theoremstyle{remark}

\newtheorem{remark}{Remark}
\endlocaldefs

\begin{document}

\begin{frontmatter}
\title{Simulation-consistent Estimation of the Marginal Likelihood for Block Models}
\runtitle{Simulation-consistent Marginal Likelihood Estimation for Block Models}

\begin{aug}
\author[A]{\fnms{Martin}~\snm{Metodiev}\ead[label=e1]{martin.metodiev@doctorant.uca.fr}\orcid{0009-0000-9432-3756}},
\author[B]{\fnms{Marie}~\snm{Perrot-Dockès}\ead[label=e2]{marie.perrot-dockees@u-paris.fr}},
\author[C]{\fnms{Guilhem}~\snm{Fouetillou}\ead[label=e3]{guilhem.fouetillou@sciencespo.fr}},
\author[A,D,F]{\fnms{Pierre}~\snm{Latouche}\ead[label=e4]{pierre.latouche@uca.fr}}
\and
\author[E,F]{\fnms{Adrian E.}~\snm{Raftery}\ead[label=e5]{raftery@uw.edu}}
\address[A]{Laboratoire de Mathématiques Blaise Pascal,
CNRS, Université Clermont Auvergne, \\\printead[presep={\ }]{e1,e4}}

\address[B]{MAP5, CNRS, F-75006 Paris,
Université Paris Cité\printead[presep={,\ }]{e2}}

\address[C]{Medialab,
Sciences Po Paris\printead[presep={,\ }]{e3}}

\address[D]{Institut Universitaire de France}

\address[E]{Departments of Statistics and Sociology,
University of Washington\printead[presep={,\ }]{e5}}

\address[F]{Equal contribution}

\runauthor{M. Metodiev et al.}
\end{aug}

\begin{abstract}
We propose a methodology for computing marginal likelihoods for block models. The proposed estimator computes the marginal likelihood from Markov chain Monte Carlo (MCMC) samples and is simulation-consistent, even when the size of the dataset is fixed. Moreover, it is asymptotically normal, of finite variance, invariant to label switching and can be computed efficiently, even for models with an arbitrarily large number of components. We evaluate the method through simulation studies in settings where the true marginal likelihood is available analytically. Finally, we apply the approach to a social network dataset based on the 2023 United Nations Climate Change Conference (COP28) and discuss the resulting insights.
\end{abstract}



\end{frontmatter}

\section{Introduction}

Networks, which are common across scientific disciplines, are often analysed by looking for hidden structures or connectivity patterns. The stochastic block model \citep[SBM,][]{Ho_et_al83-SBM_introduction,wang_stochastic_1987,nowicki_estimation_2001} is a random graph model for node clustering. It can handle heterogeneous structures such as communities or star patterns \citep{daudin_mixture_2008}. However, classical methods, such as likelihood-based methods, cannot be applied for inference in the SBM. Indeed, it is generally not feasible to compute the likelihood function of the SBM, since it involves a sum over all possible cluster configurations, which grows exponentially with the size of the network.

One major unsolved problem for the SBM is the estimation of the number of clusters. Many methods have been proposed \citep[see][for reviews]{Ab18-SBMreview,lee_review_2019,Pe19-SBM_review}, but it is unclear which of these should be used. An optimal choice from a Bayesian point of view is the most likely model a posteriori, namely the model that maximises the product of the marginal likelihood and the prior model probability. This minimises the model selection error rate. In this paper, we show how to estimate the marginal likelihood. To the best of our knowledge, this paper is the first to present an estimator that converges to the marginal likelihood with an increasing number of simulations, even when the size of the network is small.

 More precisely, if $Z(G)$ denotes the marginal likelihood and a prior distribution $\Pi(G)$ is given for the number of clusters, the value of $G$ that maximises the product $Z(G)\Pi(G)$ minimises the model selection error rate on average over $\Pi$ \citep{jeffreys_theory_1961}. Indeed, it has been shown that, under a specific prior on the number of clusters, it leads to a strongly consistent selection of this number. This is true in both dense and sparse regimes and even if the number of clusters grows as the size of the network goes to infinity \citep{cerqueira2020estimation}. However, the marginal likelihood is itself an integral and a sum over the likelihood function times the prior distribution, and hence not tractable computationally, necessitating estimation.

The only marginal likelihood estimators that have been used to estimate the number of clusters of the SBM that we are aware of are based on asymptotic or variational approximations \citep[see, e.g.,][]{hofman_bayesian_2008,latouche_variational_2012,La_et_al14-BayesOverlappingSBM,aicher_learning_2015,yan_bayesian_2016,peixoto_nonparametric_2017}. Even though such assumptions are reasonable in the limiting case in which the size of the network grows \citep{celisse_consistency_2012, mariadassou_convergence_2015}, there are no consistency results about variational approximations that apply to finite data sizes, to the best of our knowledge. One notion of consistency that is important for marginal likelihood estimators is the notion that they should be simulation-consistent, that is, that they should be able to compute the marginal likelihood up to an arbitrary degree of precision via an increasing number of simulations. While we know of no case of simulation-consistent marginal likelihood estimators being used to estimate the number of clusters of the SBM, there is a large literature on such estimators for other models \citep[see][for reviews]{FrWy12-review_marglikestims,llorente_marginal_2023,Li_et_al25-marglikestimreview}. Usually, these estimators are based on Markov chain Monte Carlo (MCMC) samples from the posterior distribution, which can in principle also be obtained for the SBM via Gibbs sampling \citep{nowicki_estimation_2001}. However, most of these estimation techniques also require the evaluation of the likelihood function, and thus cannot be used for the SBM.

This paper introduces the truncated harmonic mean estimator (THAMES) for the SBM. The THAMES is a simulation-consistent estimator of the marginal likelihood that does not require the evaluation of the likelihood function. In addition, it allows the computation of approximate confidence intervals for the marginal likelihood, since it is asymptotically normal and of finite variance, as the amount of simulations increases. Finally, the THAMES is symmetric, which makes it invariant to the problem of label switching, a central problem within the Bayesian MCMC estimation of SBM parameters, which is explained in the next section of this article. A R package implementing the THAMES is available \citep{thamesblock_rpackage}.

The THAMES can also be applied to the conceptually very similar problem of model-based co-clustering, where the latent block model \citep[LBM,][]{govaert_clustering_2003} is the model of choice. Finding the right number of clusters is also a particularly challenging issue in this setting \citep[see][for a survey]{biernacki_survey_2023}. While this article focuses on the SBM, the construction of the THAMES is analogous for the LBM. A definition of the THAMES for the LBM is given in the supplementary material \citep{thamesblock-supp}.

An example where estimating the marginal likelihood is required in a large but sparse network is given by a social network dataset related to the 2023 United Nations Climate Change Conference (COP28). This consists of over 4 million posts on the social network X, with over a million nodes and two million edges that were collected for this article. After reducing the size of the network by preprocessing, the THAMES is used to determine and interpret the number of clusters within this network. This analysis of the COP28 dataset is of practical relevance because it provides insight into different clusters of users formed around a major international climate event on a widely used social media platform. Thus, it can help identify patterns of online engagement and the structure of different groups in the diffusion of climate-related information.

The article is structured as follows. First, we introduce the SBM (Section \ref{sec: Marginal likelihood estimation for the stochastic block model}), and then define the THAMES marginal likelihood estimator (Section \ref{sec: Truncated harmonic mean estimation for the stochastic block model}). The THAMES is compared to state-of-the-art estimators in simulation settings where the exact marginal likelihood is available analytically (Section \ref{sec: Experiments}). It is then applied to the COP28 dataset (Section \ref{sec: The COP28 dataset}). A discussion of the results is in Section \ref{sec: discussion}.

\section{Marginal likelihood estimation for the stochastic block model}\label{sec: Marginal likelihood estimation for the stochastic block model}

\paragraph{\textbf{The stochastic block model}}

Let $Y$ denote the dataset, which is a $n\times n$ binary adjacency matrix. The SBM is defined in two steps: first, $n$ unknown allocation vectors $C_1,\dots,C_n$ are drawn according to a multinomial distribution with parameter $\tau$, \begin{align*}C_i|G,\tau\stackrel{\text{i.i.d}}\sim\text{Multinom}_G(1,\tau),\quad\sum_{g=1}^GC_{i,g}=1,\quad\forall i=1,\dots,n.
\end{align*} The allocation vector matrix $C=(C_1,\dots,C_n)$ is a $n\times G$ matrix, where $G\leq n$ refers to the number of clusters. Then, each entry $Y_{i,j}$ of $Y$ is drawn independently according to a Bernoulli distribution with unknown parameter chosen from the $G\times G$ matrix of connection probabilities $\mu$,\footnote{This defines the adjacency matrix of a binary, directed graph with no self-loops. The same methodology can be applied to the undirected version of the SBM, with minor changes.}\begin{align*}Y_{i,j}|G,\mu,C_{i,g}C_{j,k}=1\sim\text{Ber}(\mu_{g,k}),\quad \forall(i,j,g,k),i\neq j.
\end{align*}

Gibbs sampling is possible if the following prior distribution is used: \begin{align}
&\tau|G\sim\text{Dirichlet}(G,n_1^0,\dots,n_G^0),\quad\tau\in\mathbb{R}^G,\label{eq: dirichlet_prior}\\
&\mu_{g,k}|G\sim\text{Beta}(a^0_{g,k},b^0_{g,k}),\quad \mu_{g,k}\in(0,1),\quad\forall (g,k).\label{eq: beta_prior}
\end{align} Gibbs sampling \citep{nowicki_estimation_2001} can then be performed by initialising a value of $C$ (we do this by applying a $k-$means algorithm on the rows of $Y$) and then alternating between one draw from $(\mu,\tau)|C,Y$ and $n$ draws from $C_i$ conditional on $Y,\tau,\mu$ and $C_{1},\dots,C_{i-1},C_{i+1},\dots,C_n$. After a burn-in, the resulting MCMC sample $(\tau^{(1)},\mu^{(1)},C^{(1)}),\dots,(\tau^{(T)},\mu^{(T)},C^{(T)})$ of size $T$ approximately follows the posterior distribution of $(\tau,\mu,C)$. While Gibbs sampling can in principle be performed for any value of the hyperparameters $(a^0,b^0,n^0)$, we consider only values for which the prior distribution is symmetric, in the sense that the distribution of $C$ is not affected by the labels of the parameters. This leads to the problem of label switching that needs to be properly addressed.

\paragraph{\textbf{Label switching}}

Label switching refers to a phenomenon occurring in Gibbs sampling in which parameters simulated from the posterior distribution randomly switch their labels. It results from the fact that different sets of labels lead to different parameters, but do not change the distribution of $Y$ given these parameters. One main difficulty arising from label switching is the fact that ergodic estimators, such as the arithmetic mean of the MCMC sample, can exhibit high variance when applied to a sample with switched labels due to the inherent multimodality of such a sample. In this article, this problem is bypassed by a relabelling algorithm.

A relabelled sample $C^{(1)\star},\dots,C^{(T)\star}$, whose labels are not switched, can be obtained by applying the ECR algorithm \citep{PaIl10-ECR_algo} to $C^{(1)},\dots,C^{(T)}$, which relabels $C^{(1)},\dots,C^{(T)}$ such that the distribution of $C^{(1)\star},\dots,C^{(T)\star}$ is approximately equal to the distribution of a non-symmetric distribution. While originally introduced for mixture models, it is also applicable to the SBM. Indeed, there is a variety of relabelling algorithms for mixture models \citep[see, e.g.,][]{stephens_bayesian_1997,celeux_computational_2000,Ma_et_al05-allocationRelabelling,nobile_bayesian_2007,RoWa14-relabellingAlgo} that could be used for the SBM. We selected this algorithm because it is available in a R package and because it acts only on $C^{(1)},\dots,C^{(T)}$, not $\tau^{(1)},\dots,\tau^{(T)}$ or $\mu^{(1)},\dots,\mu^{(T)}$. While the latter samples are available as a byproduct of Gibbs sampling, they are not necessary to estimate the marginal likelihood, since they can be integrated out, resulting in the collapsed prior and likelihood function.

\paragraph{\textbf{Collapsed distributions of the SBM}}

The collapsed distribution of $C$ and $Y$ is defined via the collapsed prior $p(C|G)$ and the collapsed likelihood function $p(Y|G,C)$, for which the following expressions can be derived \citep[see, e.g.,][]{come_model_2015}: \begin{align*}
        p(C|G)&=\int p(C|G,\tau)p(\tau)\; d\tau=\frac{\Gamma(\sum^G_{g=1}n_{g}^0)\prod^G_{g=1}\Gamma(n_{g})}{\Gamma(\sum^G_{g=1}n_{g})\prod^G_{g=1}\Gamma(n_{g}^0)},\\p(Y|G,C)&=\int p(Y|G,\mu,C)p(\mu)d\mu=\prod_{k,g}^G\frac{\Gamma(a^0_{g,k}+b^0_{g,k})\Gamma(a_{g,k})\Gamma(b_{g,k})}{\Gamma(a_{g,k}+b_{g,k})\Gamma(a^0_{g,k})\Gamma(b^0_{g,k})},
    \end{align*} where $n_g=n_g^0+\sum^n_{i=1}C_{i,g}$, $a_{g,k}=a^0_{g,k}+\sum^n_{i=1}\sum_{j\neq i}C_{i,g}C_{j,k}Y_{i,j}$, as well as $ b_{g,k}=b^0_{g,k}+\sum^n_{i=1}\sum_{j\neq i}C_{i,g}C_{j,k}(1-Y_{i,j})$ and $\Gamma$ is the gamma function. The simple analytic form of this expression allows us to easily evaluate these functions on the MCMC sample $C^{(1)},\dots,C^{(T)}$, which is the basis of our estimation procedure. 
    
    The purpose of this article is the estimation of\begin{align}\label{eq: marglik_definition}
        Z(G)=\sum_{C}p(C|G)p(Y|G,C),
    \end{align}the marginal likelihood of a stochastic block model with $G$ clusters. We now  survey previous work which could be applied in this context.

\paragraph{\textbf{Previous work}}

While we are not aware of any simulation-consistent marginal likelihood estimators that were used to estimate the number of clusters of the SBM in the literature, we did find marginal likelihood estimators for other, similar purposes, that could be used for the SBM. All of these estimators can be described in the framework of reciprocal importance sampling \citep[RIS,][]{GelfandDey1994}.

For any density $h$, a RIS estimator is defined as \begin{align}\label{eq: RIS_SBM}
    \hat{Z}^{-1}_{\text{RIS}}(G)=\frac{1}{T}\sum^T_{t=1}\frac{h(C^{(t)})}{p(Y|G,C^{(t)})p(C^{(t)}|G)},
\end{align} where $C^{(1)},\dots,C^{(T)}$ are the MCMC draws from the posterior. 

RIS estimators are simulation-consistent and asymptotically normal under standard conditions, due to the following theorem.\footnote{The proof of this and any other theorem in this manuscript is given in the supplementary material \citep{thamesblock-supp}.}

\begin{theorem}\label{ref: RIS_asymnormal}
Suppose that $h$ is a function such that $\sum_Ch(C)=1.$ If $C^{(1)},\dots,C^{(T)}$ is an ergodic Markov chain with limiting distribution $p(C|G,Y)$, the posterior distribution of $C$ given $Y$ and $G$, then $\hat{Z}^{-1}_{\text{RIS}}(G)$ almost surely converges to $Z^{-1}(G)$. If $C^{(1)},\dots,C^{(T)}$ is additionally ergodic of degree 2, then $\sqrt{T}(\hat{Z}^{-1}_{\textup{RIS}}(G)-Z^{-1}(G))$ converges in distribution to a centred normal distribution.
\end{theorem}

However, RIS estimators are ergodic estimators, and hence susceptible to the problem of label switching. This problem disappears if these estimators are symmetric (i.e., invariant to changes of the labels of $C^{(1)},\dots,C^{(T)}$), since the value of these estimators does not change when label switching occurs.

One of the simplest symmetric estimators that can be derived from Equation \eqref{eq: RIS_SBM} is the harmonic mean estimator \citep{NewtonRaftery1994}, \begin{align}\label{eq: harmonicmean_SBM}
    \hat{Z}^{-1}_{\text{HME}}(G)=\frac{1}{T}\sum^T_{t=1}\frac{1}{p(Y|G,C^{(t)})}.
\end{align} It results from Equation \eqref{eq: RIS_SBM} when choosing $h(C)=p(C|G)$. A version of this estimator was used by \citet{Le_et_al22-harmonicmeanSBM} to estimate the marginal likelihood in the infinite relational model. However, while Equation \eqref{eq: harmonicmean_SBM} is simple and easy to compute, the harmonic mean estimator can be very unstable, as noted by the authors \citep{NewtonRaftery1994}. It is shown in the simulation section of this article that this lack of stability also occurs in the case of the SBM.

A more stable symmetric estimator can be obtained by focusing on a maximum a posteriori (MAP) estimator $\hat{C}_{\text{MAP}}$, which is an allocation vector matrix within the set \\$\{C^{(1)},\dots,C^{(T)}\}$ which maximises the unnormalised posterior probability, which is the product of the collapsed prior and the collapsed likelihood, $p(\hat{C}_{\text{MAP}}|G)p(Y|G,\hat{C}_{\text{MAP}}).$ Let $A_{\text{MAP}}$ denote the set that contains $\hat{C}_{\text{MAP}}$ and all of its label switched versions. Choosing $h$ to be the uniform distribution on $A_{\text{MAP}}$ results in the ChibPartition estimator \begin{align}\label{eq: ChibPartition}
    \hat{Z}^{-1}_{\text{ChibPartition}}(G)=\frac{1}{T}\sum^T_{\substack{t=1,\\C^{(t)}\in A_{\text{MAP}}}}\frac{1/G!}{p(C^{(t)}|G)p(Y|G,C^{(t)})}
\end{align} as suggested by \citet{hairault_evidence_2022}\footnote{The expression that we give for this estimator is equivalent to the one given by \citet{hairault_evidence_2022}, as shown in the supplementary material \citep{thamesblock-supp}.} in the case of mixture models, who also showed how their estimator can be computed in time $O(T)$. A version of this identity was also used by \citet{mcdaid_clustering_2012} in the case of the collapsed stochastic block model in which $G$ was not fixed, but random and sampled via a MCMC algorithm.

While extremely stable when the posterior probability of $A_{\text{MAP}}$ is close to one, the performance of the ChibPartition estimator deteriorates when this probability approaches zero. In this case, most terms in the above sum are 0, and the remaining terms are very large, since they are proportional to the reciprocal of the posterior probability of $A_{\text{MAP}}$, thus resulting in a high variance. In the COP28 dataset we considered, the sample from the posterior distribution had \textbf{no repetitions}, thus implying an extremely small posterior probability of $A_{\text{MAP}}$. A choice of $h$ that is more adapted to the posterior probabilities of the allocation vectors can be constructed by a variational approximation.

\citet{hajargasht_accurate_2020} suggested setting $h$ equal to the approximate posterior distribution of the parameters, derived in a variational setting which is similar to the setting of the SBM. There does exist a variational approximation of the posterior distribution of $C$, which is given by a multinomial distribution \citep{latouche_variational_2012}. Thus, one possible choice would be \begin{align}\label{eq: variational_h}
    \tilde{h}(C)=\prod^{n}_{i=1}\text{Multinom}_G(C_i;1,\hat{z}_i),
\end{align}where the posterior distribution $\hat{z}_i$ of $C_i$ is estimated from the MCMC sample. This can be done by splitting the MCMC sample, such that $\hat{z}_{i}$ is computed via an estimator that acts on the first half of the relabelled sample. This estimator, derived from Laplace's succession rule \citep{robert2007bayesian}, is given by\begin{align}\label{eq: variational_z}
    \hat{z}_{i,g}=\frac{1+\sum_{t=1}^{T/2}C_{i,g}^{(t)\star}}{2+T/2},
\end{align}where $C^{(1)\star},\dots,C^{(T)\star}$ denote the relabelled MCMC draws, while the marginal likelihood estimator is computed on the second half. However, this estimator is subject to label switching, since $\tilde{h}$ is not symmetric. Symmetry can be restored by averaging over all possible permutations $P_1,\dots,P_{G!}$ of the labels of $C^{(1)},\dots,C^{(T)}$, as suggested by \citet{Be_et_al03-quick_marglikestim_01} in the case of mixture models. This follows from the fact that changing the order of a sum does not change its value, and hence switching the labels of a permutation-averaged estimator does not change its value either. For example, adding one estimator computed on the label (1,2) to the same estimator computed on the label (2,1) is the same as adding this estimator computed on (2,1) to itself computed on (1,2). Mathematically, the permutations can be expressed as\begin{align*}
    P_o(C)_{i,g}=\begin{cases}1&C_{i,S_{o,g}}=1,\\0&\text{else,}\end{cases}\quad\forall o=1,\dots,G!,\forall i=1,\dots,n,\forall g=1,\dots,G,
\end{align*} with $S=(S_1,\dots,S_{G!})$ denoting a matrix of all different orderings of the vector $(1,\dots,G)$. Restoring symmetry by averaging results in the variational Bayes reciprocal importance sampling (RISVB) estimator, \begin{align}\label{eq: RISVB_estimator}
    \hat{Z}^{-1}_{\text{RISVB}}(G)=\frac{1}{G!}\sum^{G!}_{o=1}\frac{1}{T/2}\sum^{T}_{t=T/2+1}\frac{\tilde{h}(P_o(C^{(t)}))}{p(C^{(t)}|G)p(Y|G,C^{(t)})}.
\end{align} This estimator has the added benefit of using the variational approximation of the SBM, which is known to be quite accurate, while still being simulation-consistent even when this variational approximation does not hold exactly. However, the RISVB estimator is computationally costly, as the number of possible permutations grows superexponentially with the number of clusters $G$. Even a medium number of clusters such as $G=15$ leads to $15$ factorial (over 1 trillion) permutations. Hence, this estimator is unsuitable for the COP28 dataset, where the number of clusters may well surpass 15. 

An estimator that computes a sum over all permutation functions efficiently even when the number of clusters is large  was presented in \citeauthor{metodiev2025mixturemodels} ({\color{blue}2025a}) for mixture models. The central idea of this estimator was to use results based on the normal approximation to the posterior, presented in \citeauthor{metodiev_easily_2025} ({\color{blue}2025b}). However, this estimator is not appropriate for the SBM because it requires the evaluation of the untractable likelihood function and because the normal approximation of the posterior is not applicable to the posterior distribution of $C$ due to its discrete nature. Therefore, the work of \citeauthor{metodiev2025mixturemodels} ({\color{blue}2025a}) is of no interest in the SBM context.

\section{Truncated harmonic mean estimation for the stochastic block model}\label{sec: Truncated harmonic mean estimation for the stochastic block model}
    
The solution found and presented in this article consists in combining the uniform density of the ChibPartition estimator with the variational approximation used to construct the RISVB estimator to obtain the truncated harmonic mean estimator (THAMES) for the SBM: \begin{align*}
    \hat{Z}_{\text{THAMES}}^{-1}(G)=\frac{1}{G!}\sum^{G!}_{o=1}\frac{1}{T/2}\sum^T_{\substack{t=T/2+1\\P_o(C^{(t)})\in A}}\frac{1/|A|}{p(C^{(t)}|G)p(Y|G,C^{(t)})}.
\end{align*} Here, $A$ denotes a finite set and $|A|$ denotes its size. 

The THAMES is \begin{itemize}
    \item simulation consistent and asymptotically normal under standard conditions (Theorem \ref{ref: RIS_asymnormal}),
    
    \item symmetric and thus unaffected by label switching,
    
    \item and efficiently computed, using techniques presented later.
\end{itemize}

The truncation set $A$ that defines the THAMES is chosen via the variational approximation used in the definition of the RISVB. In the remainder of this section, it is first shown how to compute such an optimal $A$. Then, it is explained how the shape of $A$ allows for efficient computation of the THAMES. Finally, the special case of empty clusters is addressed. A pseudo-code algorithm detailing the implementation of the THAMES is given in the supplementary material \citep{thamesblock-supp}.

\subsection{Choosing $A$}\label{ssec: Choosing $A$}

Due to the following theorem, out of all symmetric choices of $A$, the variance of the THAMES is minimised by a highest-posterior-density (HPD) region \begin{align*}
    H_{\alpha}=\{C:p(Y|G,C)p(C|G)>q_\alpha\},
\end{align*} under fairly mild conditions, where $q_\alpha\in\mathbb{R}$ is the largest constant for which $H_\alpha$ has a posterior probability of at least $\alpha$.

\begin{theorem}
Let $C^{(T/2+1)},\dots,C^{(T)}\sim p(C|G,Y)$ be an independent sample from the posterior distribution. If no two allocation vector matrices have the same probability (up to permutations), meaning that for any two allocation vector matrices $C^1,C^2$, \begin{align*}
    p(C^1|G,Y)=p(C^2|G,Y)\Rightarrow\exists o\in\{1,\dots,G!\}:P_o(C^1)=C^2,
\end{align*} then there exists an $\alpha\in(0,1]$ such that the variance of $\hat{Z}_{\textup{THAMES}}^{-1}(G)$, conditional on $Y$, is minimised over all symmetric sets by $H_\alpha$. A set $A$ is symmetric if, and only if the image $P_o(A)=A$ for all $o\in\{1,\dots,G!\}$.
\end{theorem}

\begin{remark}
The main conditions on the distribution of $C^{(T/2+1)},\dots,C^{(T)}$ are the condition that $C^{(T/2+1)},\dots,C^{(T)}$ independently follow the posterior distribution and that there are no ties, in the sense that no two values have the same probability (up to a permutation of the labels, which do not change this probability). This first condition is usually not satisfied exactly, though it may be satisfied approximately when a large burn-in is used and when the sample from the posterior is heavily thinned. The second condition depends on the dataset, and we conjecture that it is usually satisfied for sufficiently large and/or complex datasets. It should, however, be emphasised that these two conditions are not needed to prove simulation-consistency and asymptotic normality of the THAMES, which follows from Theorem \ref{ref: RIS_asymnormal}.
\end{remark}

One can check if $C^{(t)}\in H_\alpha$ by computing $p(Y|G,C^{(t)})p(C^{(t)}|G)$. The value of $q_\alpha$ can be found approximately by taking the empirical $(1-\alpha)$-quantile of $p(Y|G,C^{(t)})p(C^{(t)}|G)$, $1\leq t\leq T$. However, the size of $H_\alpha$ is needed to compute the THAMES, and we know of no simple way to compute this size in general.

A set whose size is easily computed is given by the Cartesian product\begin{align*}
    E_{\hat{C},\hat{r}}=\{\hat{C}_1^{(1)},\dots,\hat{C}_1^{(\hat{r})}\}\times\dots\times\{\hat{C}_{n}^{(1)},\dots,\hat{C}_n^{(\hat{r})}\},\quad 1\leq \hat{r}\leq T/2.
\end{align*} To approximate a HPD region, $\hat{C}=(\hat{C}^{(1)},\dots,\hat{C}^{(\hat{r})})$ are chosen from the first half of the relabelled MCMC sample based on their posterior probabilities $p(Y|G,\hat{C}^{(t)})p(\hat{C}^{(t)}|G)$, $1\leq t\leq \hat{r}$. If the variational assumptions hold, i.e., if the posterior distribution is defined by the parameters $\hat{z}$ via Equation \eqref{eq: variational_h}, and under some additional assumptions on the concentration of $\hat{z}$ and $E_{\hat{C},\hat{r}}$, the variance of the THAMES can easily be minimised with respect to $\hat{r}$ when $E_{\hat{C},\hat{r}}$ is chosen as a truncation set, as shown in the supplementary material \citep{thamesblock-supp}.

The specific shape of $E_{\hat{C},\hat{r}}$ is convenient because the number of its elements is easily computed by taking the product of the number of elements of each set. However, this set can be suboptimal, for example when the variational assumptions do not hold. To remedy this, $E_{\hat{C},\hat{r}}$ is combined with the HPD region by taking the intersection $B_{\hat{C},\hat{r},\alpha}=E_{\hat{C},\hat{r}}\cap H_\alpha$. This set is more robust to violations of the variational assumptions, since it approximates $H_\alpha$ and the optimality results which we showed with regards to $H_\alpha$ do not rely on these assumptions. The size of $B_{\hat{C},\hat{r},\alpha}$ is also easily approximated via a Monte Carlo approach, by independently sampling $\eta^{(1)},\dots,\eta^{(T)}$ from the uniform distribution on $E_{\hat{C},\hat{r}}$ and by calculating\begin{align*}
    |A|\simeq\widewedge{|B_{\hat C,\hat r,\alpha}|}=\frac{1}{T}\sum_{\substack{t=1\\\eta^{(t)}\in H_\alpha}}^T|E_{\hat{C},\hat{r}}|.
\end{align*}Thus $A=B_{\hat{C},\hat{r},\alpha}$ is used to define the THAMES for the SBM. Its control-parameter $\alpha$ is chosen empirically, by selecting a grid on $(0,1]$, computing the sample variance of the THAMES for each value of the grid, and by setting $\alpha$ such that the variance corresponding to its grid value is minimal \citep[see the supplementary material,][for details on the grid selection]{thamesblock-supp}. Note that, for a fixed value of $\alpha$, the THAMES is equal to a RIS estimator multiplied by $\frac{\widehat{|B_{\hat{C},\hat{r},\alpha}|}}{|B_{\hat{C},\hat{r},\alpha}|}$, which almost surely converges to a constant. Hence, simulation-consistency and asymptotic normality still hold by Slutsky's theorem. 

\subsection{Calculating the THAMES for the stochastic block model efficiently\label{ssec: Calculating the THAMES for the stochastic block model efficiently}}

In \citeauthor{metodiev2025mixturemodels} ({\color{blue}2025a}), it was shown that sums over permutations of a truncated estimator can be computed efficiently if an appropriate ordering of the parameter space is used. This approach will be extended to allocation vectors by identifying a lexicographical ordering of the columns of any allocation vector matrix $C$ from the second half of the MCMC sample, $C\in\{C^{(T/2+1)},\dots,C^{(T)}\}$. Since there are $G$ columns of $C$, it suffices to choose $G$ rows to define such an ordering. Let $\upsilon=\{\upsilon_1,\dots,\upsilon_G\}$ be the set of size $G$ that defines these rows and let $\psi_\upsilon(C)$ denote a lexicographical ordering, in the sense that it is the unique permutation of the columns of $C$ for which \begin{align}\label{eq: definition_ordering_constraint}
    \min\{i\in\upsilon:\psi_\upsilon(C)_{i,1}=1\}>\dots> \min\{i\in\upsilon:\psi_\upsilon(C)_{i,G}=1\}.
\end{align} For example, in the case that $G=2$, $n=2$ and $\upsilon=\{1,2\}$, \begin{align*}
    (\psi_\upsilon(C)_{i,1},\psi_\upsilon(C)_{i,2})=\begin{cases}(C_{i,1},C_{i,2})&\min\{i:C_{i,1}=1\}> \min\{i:C_{i,2}=1\}\\(C_{i,2},C_{i,1})&\text{else.}\end{cases}
\end{align*}Let $\Omega$ be equal to an index set that includes all permutations for which $P_o(\psi_\upsilon(C^{(t)}))\in B_{\hat{C},\hat{r},\alpha}$ for some value of $T/2+1\leq t\leq T$. Due to the following theorem, this set can be used to drastically reduce computation time.

\begin{theorem}
If there are no empty clusters on the second half of the MCMC sample, in the sense that the event $\sum_{i=1}^n C_{i,g}^{(t)}=0$ does not occur for any value of $T/2+1\leq t\leq T$ and $1\leq g\leq G$, the THAMES is equal to \begin{align}\label{eq: thames_with_omegae}
    \hat{Z}^{-1}_{\textup{THAMES}}(G)=\frac{1}{G!}\sum_{o\in\Omega}\frac{1}{T/2}\sum^T_{\substack{t=T/2+1\\P_o(\psi_\upsilon(C^{(t)}))\in A}}\frac{1/|A|}{p(C^{(t)}|G)p(Y|G,C^{(t)})}.
\end{align}
\end{theorem}

\begin{remark}
Equation \eqref{eq: thames_with_omegae} simplifies the computation of the THAMES, since it implies that only a sum over $\Omega$ needs to be computed, whose size is usually equal to one in our experience, instead of a sum over $\{1,\dots,G!\}$, whose size grows superexponentially with $G$. 

The small size of $\Omega$ can be explained as follows: while there may be a certain amount of uncertainty concerning the cluster assignments of some nodes, most clusters have the property that there exists at
least one node to which they are assigned with posterior
probability close to one. We need only sum over $\tilde{G}$! many permutations, where
$\tilde{G}$ denotes the number of clusters which do not have this
property. More precisely, setting $\upsilon_g=\textup{argmax}_i \hat{z}_{i,g}$ for all  $g\in\{1,\dots,G\}$ allows a particularly simple choice of $\Omega$. This choice is obtained by identifying in which case the indices $i\in \upsilon$ lead to a size of \begin{align*}
    |\{\hat{C}_i^{(1)},\dots,\hat{C}_i^{(\hat{r})}\}|=1.
\end{align*} If that is the case, permutations which permute column $i$ are not in $\Omega$, since in that case $P_o(\psi_{\upsilon}(C))\notin E_{\hat{C},\hat{r}}$ for any $P_o$ that is not the identity. Thus, $|\Omega|=\tilde{G}!$, where $\tilde{G}$ denotes the number of clusters $g$ for which no node $i$ exists which fulfils the above equation. We often found that $\tilde{G}=0$, even for the COP28 dataset, which is quite large. This may be the case because only 1 out of the over 10,000 datapoints needs to have a very high classification probability for the cluster index to be considered a fixed index, which is actually not a very strict condition. The only cases in which we observed larger values of $\tilde{G}$ were cases in which there were empty clusters in the second half of the MCMC sample. Luckily, the case of empty clusters can be easily sidestepped.
\end{remark}

\subsection{Dealing with empty clusters\label{ssec: Dealing with empty clusters}}

An empty cluster occurs whenever $\sum_{i=1}^n C_{i,g}^{(t)}=0$ for some value of $1\leq t\leq T$ and $1\leq g\leq G$. Let $p_0(G)$ denote the posterior probability that cluster $G$ is empty. We use the following theorem, which was proven in \citet{No04-MixtureModelsDifferentSizeLink,No07-MargLikEmptyComponents} for mixture models, but which also holds for the SBM.

\begin{theorem}

If the prior on the SBM parameters is given by Equations \eqref{eq: dirichlet_prior} and \eqref{eq: beta_prior}, then \begin{align}\label{eq: nobileformula_sbm}
    Z(G)=Z(G-1)\cdot\frac{1}{p_0(G)}\cdot\frac{\Gamma(\sum^G_{g=1}n_g^0)\Gamma(n+\sum^{G-1}_{g=1}n_g^0)}{\Gamma(\sum^{G-1}_{g=1}n_g^0)\Gamma(n+\sum^G_{g=1}n_g^0)}.
\end{align}
\end{theorem}

Whenever an empty cluster occurs, one can thus simply estimate \begin{align*}\hat{p}_0(G)=\frac{1}{G}\sum^G_{g=1}\frac{1}{T}\sum_{t=1}^T\prod_{i=1}^n (1-C_{i,g}^{(t)}) ,
\end{align*} and insert the THAMES estimate of the lower dimensional model $\hat{Z}_{\text{THAMES}}(G-1)$ into Equation \eqref{eq: nobileformula_sbm} whenever an empty cluster occurs in the MCMC sample.

\section{Experiments}\label{sec: Experiments}

\begin{table}[]
    \centering
    \caption{parameters of the three different scenarios; $\lambda$ is varied between 0.3 and 0.7 within the experiments}
    \begin{tabular}{c|c|c|c|}
     &\text{Communities}&\text{Disassortative Mixing}&\text{Disassortative Mixing}  \\
     &&\text{(symmetric)}&\text{(asymmetric)}\\\hline $\mu$&$\begin{pmatrix}
        \lambda&0.01\\0.01&\lambda
    \end{pmatrix}$&$\begin{pmatrix}
        0.01&\lambda\\\lambda&0.01
    \end{pmatrix}$&$\begin{pmatrix}
        0.01&0.01\\\lambda&0.01
    \end{pmatrix}$\\\hline$\tau$&$(1/2,1/2)$&$(1/2,1/2)$&$(1/2,1/2)$\\
\end{tabular}
    \label{tab: sim_01}
\end{table}

To evaluate the performance of the THAMES, it is tested in simulation settings in which the true value of the marginal likelihood is known. The only settings in which the exact marginal likelihood can be computed in a feasible amount of time are settings in which both the size of the data $n$ and the number of clusters $G$ are very low. To this end, $n=10$ and $G=2$ are chosen, since the sum occurring in Equation \eqref{eq: marglik_definition} has only $G^n=2^{10}=1024$ terms in this case.

In this section and the next, we choose the hyperparameters $a^0_{g,k}=b^0_{g,k}=n_1^0=\dots=n_G^0=1$, since it has been shown that the posterior distribution obtains near-optimal contraction rates in this case \citep{Gh20-SBMcontractionrates}. Three scenarios are considered, whose parameters are described in detail in Table \ref{tab: sim_01}.

\subsection{Evaluating the performance of the THAMES with increasing size $T$ of the MCMC sample}

\begin{figure}
    \centering \begin{tabular}{c}
        \includegraphics[scale=.41]{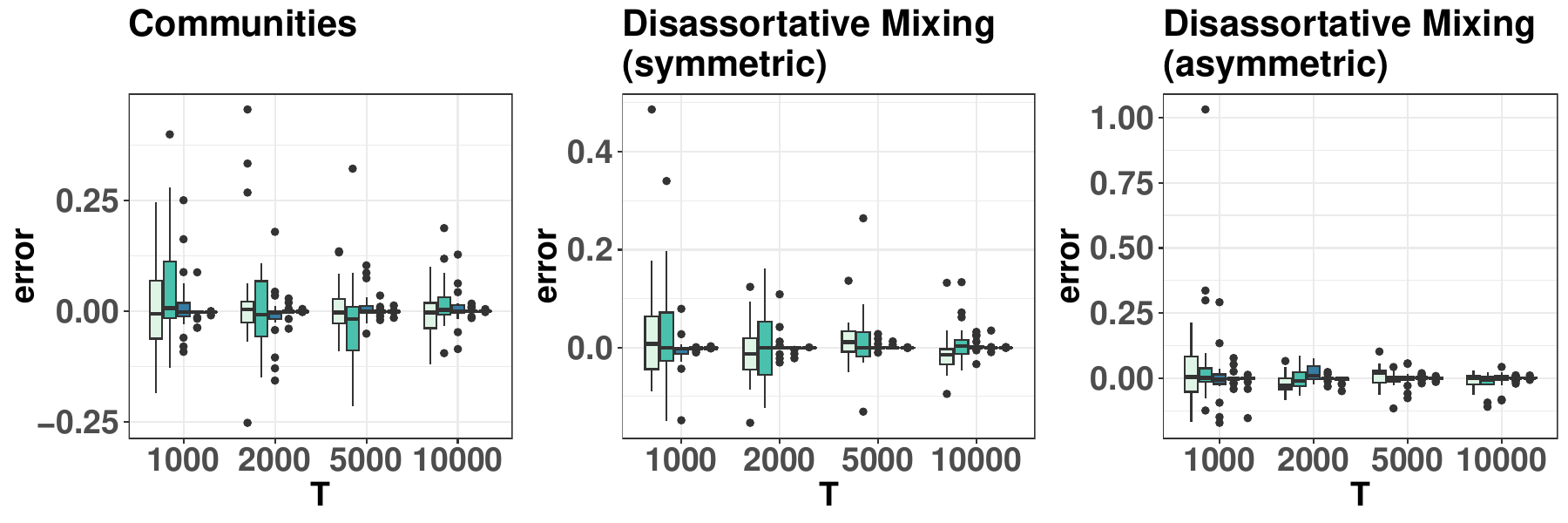}\\
         \includegraphics[scale=.6]{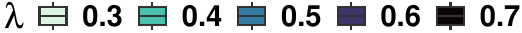}
    \end{tabular}
    
    \caption{The error $\log(Z)-\log(\hat{Z}_{\text{THAMES}})$ of the logarithm of the THAMES for 20 different datasets for each value of $\lambda$ and for different MCMC sample sizes $T$, as well as different scenarios (communities, symmetric disassortative mixing, asymmetric disassortative mixing)}
    \label{fig:sim01_results_varyingT}
\end{figure}

For each value of $\lambda\in\{0.3,0.4,0.5,0.6,0.7\}$, 20 datasets were simulated. Then, Gibbs sampling was performed with a burn-in of 2,000 to obtain a MCMC sample of size $T=1,000,T=2,000,T=5,000,$ and $T=10,000$, respectively. The THAMES was computed for each dataset. The error between the logarithm of the THAMES and the logarithm of the true marginal likelihood $\log(Z)-\log(\hat{Z}_{\text{THAMES}})$ is shown in Figure \ref{fig:sim01_results_varyingT}. The results are in line with the fact that the THAMES is simulation-consistent, since this error is decreasing as the number of simulations $T$ increases, for all simulation settings and all values of $\lambda$.

\subsection{The THAMES compared to simulation-consistent and variational estimators}

\begin{figure}
    \centering \begin{tabular}{c}
        \includegraphics[scale=.4]{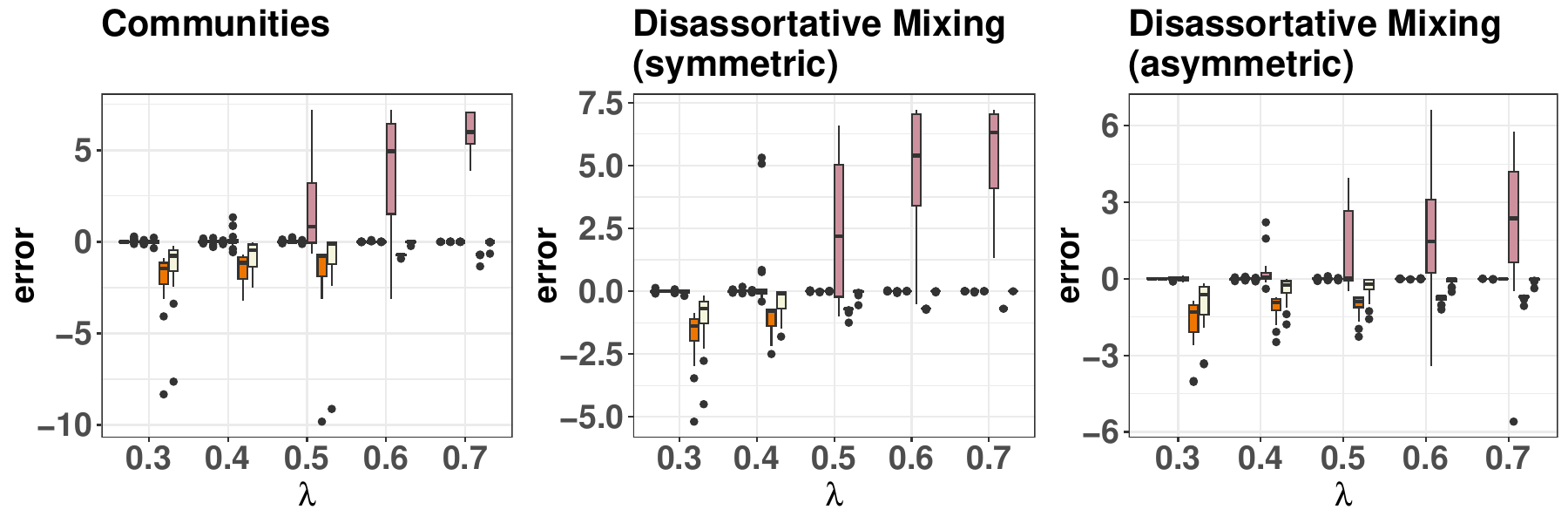}\\
         \includegraphics[scale=.6]{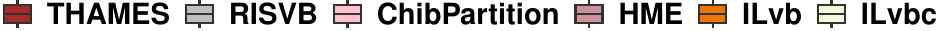}
    \end{tabular}
    
    \caption{The error $\log(Z)-\log(\hat{Z})$ of the logarithm of the different estimators (the THAMES, the RISVB estimator, the ChibPartition estimator, the harmonic mean estimator and the ILvb estimator as well as its correction, ILvbc) for 20 different datasets as well as different scenarios (communities, symmetric disassortative mixing, asymmetric disassortative mixing)}
    \label{fig:sim01_results_varyinglambda}
\end{figure}

Since we are not aware of any other simulation-consistent marginal likelihood estimators having been used in this setting, the THAMES is compared to simulation-consistent marginal likelihood estimators that were described and adapted to the SBM in this article. These are the harmonic mean estimator (HME, Equation \ref{eq: harmonicmean_SBM}), the variational Bayes reciprocal importance sampling estimator (RISVB, Equation \ref{eq: RISVB_estimator}) and the ChibPartition estimator (Equation \ref{eq: ChibPartition}). The performance of the THAMES is also compared to a variational Bayes approximation of the marginal likelihood, the ILvb \citep{latouche_variational_2012}. In simulations, we observed that the ILvb was biased by a factor roughly equal to $G!$ because the variational Bayes' algorithm is known to focus on a single mode when approximating the posterior distribution \citep{balocchi_understanding_2025}, so we also compared the THAMES to a version of ILvb for which this bias was corrected, denoted as ILvbc. We refer to the supplementary material for details on the computation of these variational estimators \citep{thamesblock-supp}.

The results are shown in Figure \ref{fig:sim01_results_varyinglambda}. Each estimator was computed on a MCMC sample of size $T=10,000$, for 20 different datasets independently simulated for each value of $\lambda\in\{0.3,0.4,0.5,0.6,0.7\}$. The error of the logarithm, $\log(Z)-\log(\hat{Z})$, was plotted for each estimator. Only the harmonic mean estimator did not converge in the case where $\lambda$ was high, likely because the variance of the harmonic mean estimator is high when the posterior distribution is very different from the prior distribution, and the posterior is much more concentrated than the prior since it is concentrated around the true solution in this case. For the variational estimators, the converse is true: while they  approximate the marginal likelihood well when $\lambda$ is large, presumably because variational assumptions hold, they are clearly biased when $\lambda$ is small, even when taking into account the $G!$ correction.

\subsection{The THAMES compared to simulation-consistent and variational estimators on one specific dataset}

\begin{figure}
    \centering \includegraphics[scale=.42]{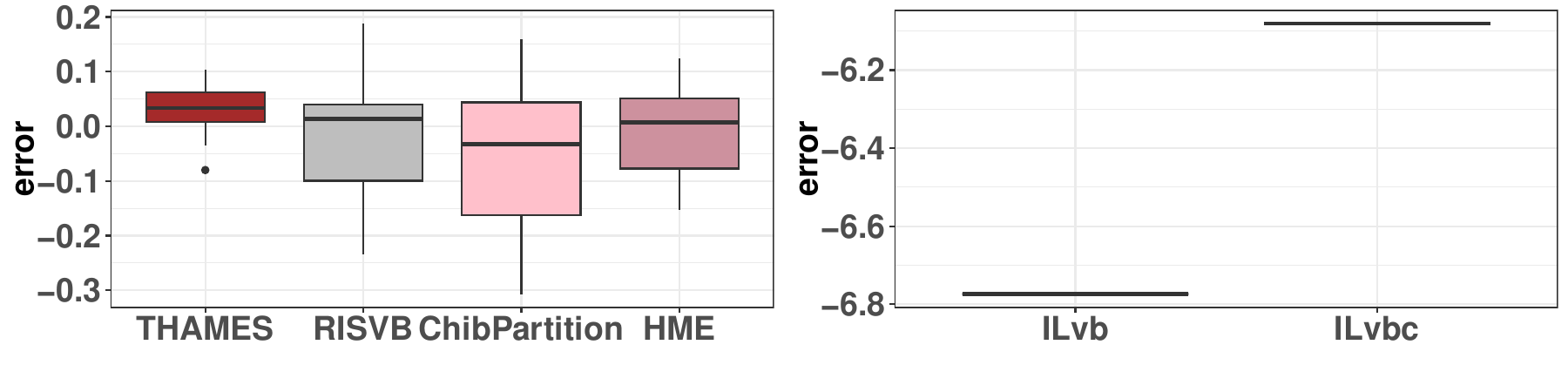}
    \caption{The error $\log(Z)-\log(\hat{Z})$ of the logarithm of the THAMES, the RISVB estimator, the ChibPartition estimator, the harmonic mean estimator and the ILvb estimator as well as its correction, ILvbc, for 20 different MCMC simulations on the same dataset; the values of ILvb and ILvbc do not change since they are not calculated on a MCMC sample}
    \label{fig:difficultgraph_results}
\end{figure}

In most simulations, we observed that the posterior probability of the maximum a posteriori estimator $\hat{C}_{\text{MAP}}$, an allocation vector matrix within $\{C^{(1)},\dots,C^{(T)}\}$ which maximises the unnormalised posterior probability $p(\hat{C}_{\text{MAP}}|G)p(Y|G,\hat{C}_{\text{MAP}})$, was particularly high since there was one value of the MCMC sample that repeated itself very often. However, in the case of the COP28 dataset the posterior probability of the MAP is so low that it does not repeat at all. To simulate this problem, we chose a dataset for which the posterior probability of the MAP is low \citep[see the supplementary material,][for details]{thamesblock-supp}. The simulation-consistent estimators were computed on 20 independent MCMC simulations based on this dataset. The error, $\log(Z)-\log(\hat{Z})$, of the logarithm of the estimators is shown in Figure \ref{fig:difficultgraph_results}. The THAMES clearly outperforms the other estimators because it is more adapted to deal with this low probability setting. The ILvb and ILvbc estimators are quite biased compared to the other estimators, with errors of -6.77 and -6.08, respectively. Their errors do not change with different MCMC simulations, since these estimators do not use these simulations to approximate the marginal likelihood. 

\section{The COP28 dataset \label{sec: The COP28 dataset}}

The COP28 dataset contains 1,182,423 nodes (users) and 2,754,066 directed edges (quotes/reposts), which were extracted from an exhaustive corpus of X data covering the full year of 2023. The corpus includes all posts and reposts containing the exact string “cop28” or “\#cop28”, regardless of capitalisation and language. Since “cop28” was the widely used label for the event, this query provides very strong coverage of the related conversations on X. It should also be noted that, at the time of data collection, Linkfluence, which provided this dataset, had full access to the entire X Firehose (GNIP), ensuring exhaustive data collection.

\begin{figure}
    \centering
    \begin{tabular}{c}
         \includegraphics[scale=.35]{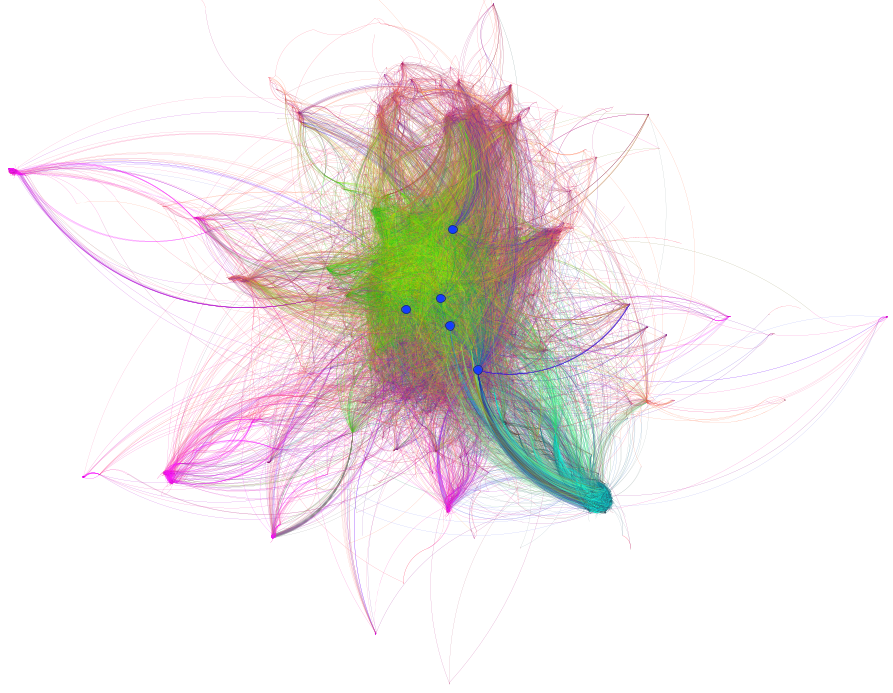}\\\includegraphics[scale=.8]{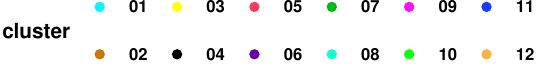}
    \end{tabular}
    
    \caption{A graphical representation of the COP28 dataset with 12 clusters, using the ForceAtlas 2 algorithm \citep{Ja_et_al14-forceatlas2} on the platform gephi \citep{Ba_et_al09-gephi}}
    \label{fig:COP28_visualization}
\end{figure}

Since the COP28 dataset is too large to allow fast computation with our available processing power, we first decreased it to 10,979 nodes and 74,826 edges by preprocessing steps explained in the supplementary material \citep{thamesblock-supp}, and then computed the THAMES from $G=2$ to $G=20$. It was maximised for $G=12$. Interestingly, this differs considerably from the value that maximises the integrated complete data log likelihood \citep[ICL,][]{biernacki_assessing_2002, daudin_mixture_2008, greedpackage}, a state-of-the-art alternative to the marginal likelihood, which estimates $G=37$ when using the greed package \citep{greedpackage}. It appears that the clusters provided by maximising the ICL are quite fractured, with one cluster even only containing one user \citep[see the supplementary material,][for a detailed comparison of these results with the results obtained by the THAMES]{thamesblock-supp}. A MAP estimate was computed on the sample with $G=12$ clusters. The resulting clustering is visualised in Figure \ref{fig:COP28_visualization}.

Cluster 11 has the interesting property that only 5 nodes (UN Climate Change, Al Gore, COP28 UAE, Loss and Damage Collaboration (L\&DC) and António Guterres)  are assigned to it and \textbf{all} other clusters tend to quote or repost this cluster. Indeed, the only clusters that had a higher probability of connecting to a cluster different from cluster 11 are clusters 1 and 4, who tend to repost cluster 8. Remarkably, there is not a single cluster for which the internal probability of quoting/reposting itself is larger than the external probability of quoting/reposting another cluster. Thus, we find the phenomenon of asymmetric disassortative mixing defined in Table \ref{tab: sim_01}, where edges run heavily from one cluster to another with no reciprocity, observed here at scale and on real data. From a sociological point of view, it might be appropriate to classify these clusters as a “public”. Publics are organised around shared attention to a common object \citep[see][for a review]{OiRi24-public_review}. In this specific context, such attention is particularly centralised onto five users. While a textual analysis may reveal a slightly more nuanced picture \citep[see the supplementary material,][]{thamesblock-supp}, the topological identification of clusters that focuses only on the repost/quote structure around the topic of the COP28 indicates a strong core-periphery structure, with all but five users at the periphery. While similar structures have been noted around the social network X \citep{Ya_et_al18-coreperiphery_twitter}, this  extreme structure seems indicative of the domination of five central users of X's landscape around the COP28.

\section{Discussion \label{sec: discussion}}

We have introduced the THAMES for block models. To the best of our knowledge, it is the first simulation-consistent marginal likelihood estimator to be used to determine the number of clusters in the SBM or LBM. It is asymptotically normal and can be efficiently computed, even on a dataset with more than 10,000 nodes.

The THAMES for block models relies on the fact that the collapsed distributions of the latent allocation vectors are available in a simple analytic expression. As such, it can be applied to any block model that fulfils this condition \citep[e.g., ][]{ke_et_al06-infinite_relational_model,WyFr12-collapsed_infinite_LBM,peixoto_nonparametric_2017,Pa_et_al26-collapsed_SBM_variations}. Interesting future work might be to create simulation-consistent marginal likelihood estimators that do not need the collapsed distribution, but that can act on the joint density of the discrete and continuous parameters, as suggested by \citet{Bi_et_al10-marglikestimLatentClassModel,tancini_marginal_2025}.

A different approach to ours would be to include a prior on the number of clusters to directly sample from the whole posterior space \citep{mcdaid_clustering_2012,newman_estimating_2016,peixoto_nonparametric_2017,Le_et_al22-harmonicmeanSBM}. Even in this approach, it is still of interest to compute the marginal likelihood for model comparison \citep{mcdaid_clustering_2012,Le_et_al22-harmonicmeanSBM}. The THAMES could be used to this end as well.

\section*{Author Contributions}
All five authors contributed to the paper.

\section*{Declarations}

\section*{Competing interests}
No competing interest is declared.

\section*{Data availability}

The dataset and the code are available on github via the following {\color{blue}\href{https://github.com/m-metodiev/thames-block-models}{link}}.

\section*{Acknowledgments}
Raftery's research was supported by the Blumstein-Jordan professorship at the University of Washington. Latouche's research was supported by the Institut Universitaire de France (IUF). The authors would like to thank Bettina Grün for helpful comments.

\begin{supplement}
\stitle{Supplement A: proofs}
\sdescription{Contains proofs for all theorems present in this article, as well as a proposition stating a formula of the variance of the THAMES in the variational case, and an expression of the ChibPartition estimator.}
\end{supplement}
\begin{supplement}
\stitle{Supplement B: details on the algorithm}
\sdescription{Contains a pseudo-code algorithm detailing the implementation of the THAMES with additional explanations, such as details on the choice of the control parameter $\alpha$, the optimal radius $r$, and details on the computation of variational approximations}
\end{supplement}
\begin{supplement}
\stitle{Supplement C: details on the datasets}
\sdescription{Contains details on the COP28 dataset and one artificial dataset used in the simulation settings.}
\end{supplement}
\begin{supplement}
\stitle{Supplement D: defining the THAMES for the latent block model}
\sdescription{Contains a definition of the THAMES for the latent block model.}
\end{supplement}

\bibliographystyle{chicago}
\bibliography{THAMES_sbm.bib}

\begin{thebibliography}{}

\bibitem[\protect\citeauthoryear{Abbe}{Abbe}{2018}]{Ab18-SBMreview}
Abbe, E. (2018).
\newblock Community detection and stochastic block models: recent developments.
\newblock {\em Journal of Machine Learning Research\/}~{\em 18\/}(177), 1--86.

\bibitem[\protect\citeauthoryear{Agrawal}{Agrawal}{2013}]{minibatchkmeans_rpackage}
Agrawal, S. (2013).
\newblock Mini-batch-k-means.
\newblock R package.
\newblock \url{https://github.com/siddharth-agrawal/Mini-Batch-K-Means}.

\bibitem[\protect\citeauthoryear{Aicher, Jacobs, and Clauset}{Aicher et~al.}{2015}]{aicher_learning_2015}
Aicher, C., A.~Z. Jacobs, and A.~Clauset (2015, June).
\newblock Learning {Latent} {Block} {Structure} in {Weighted} {Networks}.
\newblock {\em Journal of Complex Networks\/}~{\em 3\/}(2), 221--248.
\newblock arXiv:1404.0431 [stat].

\bibitem[\protect\citeauthoryear{Balocchi and Wade}{Balocchi and Wade}{2025}]{balocchi_understanding_2025}
Balocchi, C. and S.~Wade (2025).
\newblock Understanding uncertainty in {Bayesian} cluster analysis.
\newblock {\em arXiv preprint arXiv:2506.16295\/}.

\bibitem[\protect\citeauthoryear{Bastian, Heymann, and Jacomy}{Bastian et~al.}{2009}]{Ba_et_al09-gephi}
Bastian, M., S.~Heymann, and M.~Jacomy (2009).
\newblock Gephi: an open source software for exploring and manipulating networks.
\newblock In {\em Proceedings of the international AAAI conference on web and social media}, Volume~3, pp.\  361--362.

\bibitem[\protect\citeauthoryear{Berkhof, van Mechelen, and Gelman}{Berkhof et~al.}{2003}]{Be_et_al03-quick_marglikestim_01}
Berkhof, J., I.~van Mechelen, and A.~Gelman (2003).
\newblock A {B}ayesian approach to the selection and testing of mixture models.
\newblock {\em Statistica Sinica\/}~{\em 13\/}(2), 423--442.

\bibitem[\protect\citeauthoryear{Biernacki, Celeux, and Govaert}{Biernacki et~al.}{2002}]{biernacki_assessing_2002}
Biernacki, C., G.~Celeux, and G.~Govaert (2002).
\newblock Assessing a mixture model for clustering with the integrated completed likelihood.
\newblock {\em IEEE transactions on pattern analysis and machine intelligence\/}~{\em 22\/}(7), 719--725.
\newblock Publisher: IEEE.

\bibitem[\protect\citeauthoryear{Biernacki, Celeux, and Govaert}{Biernacki et~al.}{2010}]{Bi_et_al10-marglikestimLatentClassModel}
Biernacki, C., G.~Celeux, and G.~Govaert (2010).
\newblock Exact and {M}onte {C}arlo calculations of integrated likelihoods for the latent class model.
\newblock {\em Journal of Statistical Planning and Inference\/}~{\em 140\/}(11), 2991--3002.

\bibitem[\protect\citeauthoryear{Biernacki, Jacques, and Keribin}{Biernacki et~al.}{2023}]{biernacki_survey_2023}
Biernacki, C., J.~Jacques, and C.~Keribin (2023, July).
\newblock A {Survey} on {Model}-{Based} {Co}-{Clustering}: {High} {Dimension} and {Estimation} {Challenges}.
\newblock {\em Journal of Classification\/}~{\em 40\/}(2), 332--381.

\bibitem[\protect\citeauthoryear{Blei, Ng, and Jordan}{Blei et~al.}{2003}]{Bl_et_al03-LDAintroduction}
Blei, D.~M., A.~Y. Ng, and M.~I. Jordan (2003).
\newblock Latent dirichlet allocation.
\newblock {\em Journal of machine Learning research\/}~{\em 3\/}(Jan), 993--1022.

\bibitem[\protect\citeauthoryear{Bouchet-Valat}{Bouchet-Valat}{2023}]{packageSnowballC}
Bouchet-Valat, M. (2023).
\newblock {\em SnowballC: Snowball Stemmers Based on the C 'libstemmer' UTF-8 Library}.
\newblock R package version 0.7.1.

\bibitem[\protect\citeauthoryear{Celeux, Hurn, and Robert}{Celeux et~al.}{2000}]{celeux_computational_2000}
Celeux, G., M.~Hurn, and C.~P. Robert (2000).
\newblock Computational and inferential difficulties with mixture posterior distributions.
\newblock {\em Journal of the American Statistical Association\/}~{\em 95\/}(451), 957--970.
\newblock Publisher: Taylor \& Francis.

\bibitem[\protect\citeauthoryear{Celisse, Daudin, and Pierre}{Celisse et~al.}{2012}]{celisse_consistency_2012}
Celisse, A., J.-J. Daudin, and L.~Pierre (2012).
\newblock Consistency of maximum-likelihood and variational estimators in the stochastic block model.

\bibitem[\protect\citeauthoryear{Cerqueira and Leonardi}{Cerqueira and Leonardi}{2020}]{cerqueira2020estimation}
Cerqueira, A. and F.~Leonardi (2020).
\newblock Estimation of the number of communities in the stochastic block model.
\newblock {\em IEEE Transactions on Information Theory\/}~{\em 66\/}(10), 6403--6412.

\bibitem[\protect\citeauthoryear{{Côme}, {E.}, {Jouvin}, {N.}, {Latouche}, {P.}, {Bouveyron}, and {C.}}{{Côme} et~al.}{2021}]{greedpackage}
{Côme}, {E.}, {Jouvin}, {N.}, {Latouche}, {P.}, {Bouveyron}, and {C.} (2021).
\newblock Hierarchical clustering with discrete latent variable models and the integrated classification likelihood.
\newblock {\em Advances in Data Analysis and Classification\/}~{\em 15}, 957–986.

\bibitem[\protect\citeauthoryear{Côme and Latouche}{Côme and Latouche}{2015}]{come_model_2015}
Côme, E. and P.~Latouche (2015).
\newblock Model selection and clustering in stochastic block models based on the exact integrated complete data likelihood.
\newblock {\em Statistical Modelling\/}~{\em 15\/}(6), 564--589.
\newblock \_eprint: https://doi.org/10.1177/1471082X15577017.

\bibitem[\protect\citeauthoryear{Daudin, Picard, and Robin}{Daudin et~al.}{2008}]{daudin_mixture_2008}
Daudin, J.-J., F.~Picard, and S.~Robin (2008).
\newblock A mixture model for random graphs.
\newblock {\em Statistics and computing\/}~{\em 18\/}(2), 173--183.
\newblock Publisher: Springer.

\bibitem[\protect\citeauthoryear{Feinerer and Hornik}{Feinerer and Hornik}{2025}]{tmpackage}
Feinerer, I. and K.~Hornik (2025).
\newblock {\em tm: Text Mining Package}.
\newblock R package version 0.7-16.

\bibitem[\protect\citeauthoryear{Friel and Wyse}{Friel and Wyse}{2012}]{FrWy12-review_marglikestims}
Friel, N. and J.~Wyse (2012).
\newblock Estimating the evidence--a review.
\newblock {\em Statistica Neerlandica\/}~{\em 66\/}(3), 288--308.

\bibitem[\protect\citeauthoryear{Gelfand and Dey}{Gelfand and Dey}{1994}]{GelfandDey1994}
Gelfand, A.~E. and D.~K. Dey (1994).
\newblock Bayesian model choice: asymptotics and exact calculations.
\newblock {\em Journal of the Royal Statistical Society: Series B (Methodological)\/}~{\em 56}, 501--514.

\bibitem[\protect\citeauthoryear{Ghosh, Pati, and Bhattacharya}{Ghosh et~al.}{2020}]{Gh20-SBMcontractionrates}
Ghosh, P., D.~Pati, and A.~Bhattacharya (2020).
\newblock Posterior contraction rates for stochastic block models.
\newblock {\em Sankhya A\/}~{\em 82\/}(2), 448--476.

\bibitem[\protect\citeauthoryear{Govaert and Nadif}{Govaert and Nadif}{2003}]{govaert_clustering_2003}
Govaert, G. and M.~Nadif (2003, February).
\newblock Clustering with block mixture models.
\newblock {\em Pattern Recognition\/}~{\em 36\/}(2), 463--473.

\bibitem[\protect\citeauthoryear{Govaert and Nadif}{Govaert and Nadif}{2010}]{govaert_latent_2010}
Govaert, G. and M.~Nadif (2010, January).
\newblock Latent {Block} {Model} for {Contingency} {Table}.
\newblock {\em Communications in Statistics - Theory and Methods\/}~{\em 39\/}(3), 416--425.

\bibitem[\protect\citeauthoryear{Gr{\"u}n and Hornik}{Gr{\"u}n and Hornik}{2011}]{GrKu11-topicmodels_rpackage}
Gr{\"u}n, B. and K.~Hornik (2011).
\newblock topicmodels: An r package for fitting topic models.
\newblock {\em Journal of statistical software\/}~{\em 40}, 1--30.

\bibitem[\protect\citeauthoryear{Grün and Hornik}{Grün and Hornik}{2024}]{topicmodelspackage}
Grün, B. and K.~Hornik (2024).
\newblock {\em topicmodels: Topic Models}.
\newblock R package version 0.2-17.

\bibitem[\protect\citeauthoryear{Hairault, Robert, and Rousseau}{Hairault et~al.}{2022}]{hairault_evidence_2022}
Hairault, A., C.~P. Robert, and J.~Rousseau (2022).
\newblock Evidence estimation in finite and infinite mixture models and applications.
\newblock {\em arXiv preprint arXiv:2205.05416\/}.

\bibitem[\protect\citeauthoryear{Hajargasht and Woźniak}{Hajargasht and Woźniak}{2020}]{hajargasht_accurate_2020}
Hajargasht, G. and T.~Woźniak (2020, May).
\newblock Accurate {Computation} of {Marginal} {Data} {Densities} {Using} {Variational} {Bayes}.
\newblock arXiv:1805.10036 [stat].

\bibitem[\protect\citeauthoryear{Hofman and Wiggins}{Hofman and Wiggins}{2008}]{hofman_bayesian_2008}
Hofman, J.~M. and C.~H. Wiggins (2008, June).
\newblock Bayesian {Approach} to {Network} {Modularity}.
\newblock {\em Phys. Rev. Lett.\/}~{\em 100\/}(25), 258701.
\newblock Publisher: American Physical Society.

\bibitem[\protect\citeauthoryear{Holland, Laskey, and Leinhardt}{Holland et~al.}{1983}]{Ho_et_al83-SBM_introduction}
Holland, P.~W., K.~B. Laskey, and S.~Leinhardt (1983).
\newblock Stochastic blockmodels: First steps.
\newblock {\em Social networks\/}~{\em 5\/}(2), 109--137.

\bibitem[\protect\citeauthoryear{Jacomy, Venturini, Heymann, and Bastian}{Jacomy et~al.}{2014}]{Ja_et_al14-forceatlas2}
Jacomy, M., T.~Venturini, S.~Heymann, and M.~Bastian (2014).
\newblock Forceatlas2, a continuous graph layout algorithm for handy network visualization designed for the gephi software.
\newblock {\em PloS one\/}~{\em 9\/}(6), e98679.

\bibitem[\protect\citeauthoryear{Jeffreys}{Jeffreys}{1961}]{jeffreys_theory_1961}
Jeffreys, H. (1961).
\newblock {\em Theory of {Probability}\/} (3rd ed.).
\newblock Oxford, U.K.: Oxford University Press.

\bibitem[\protect\citeauthoryear{Kemp, Tenenbaum, Griffiths, Yamada, and Ueda}{Kemp et~al.}{2006}]{ke_et_al06-infinite_relational_model}
Kemp, C., J.~B. Tenenbaum, T.~L. Griffiths, T.~Yamada, and N.~Ueda (2006).
\newblock Learning systems of concepts with an infinite relational model.
\newblock In {\em AAAI}, Volume~3, pp.\ ~5.

\bibitem[\protect\citeauthoryear{Keribin, Brault, Celeux, and Govaert}{Keribin et~al.}{2015}]{keribin_estimation_2015}
Keribin, C., V.~Brault, G.~Celeux, and G.~Govaert (2015, November).
\newblock Estimation and selection for the latent block model on categorical data.
\newblock {\em Statistics and Computing\/}~{\em 25\/}(6), 1201--1216.

\bibitem[\protect\citeauthoryear{Latouche, Birmel{\'e}, and Ambroise}{Latouche et~al.}{2014}]{La_et_al14-BayesOverlappingSBM}
Latouche, P., E.~Birmel{\'e}, and C.~Ambroise (2014).
\newblock {Model selection in overlapping stochastic block models}.
\newblock {\em Electronic Journal of Statistics\/}~{\em 8\/}(1), 762 -- 794.

\bibitem[\protect\citeauthoryear{Latouche, Birmelé, and Ambroise}{Latouche et~al.}{2012}]{latouche_variational_2012}
Latouche, P., E.~Birmelé, and C.~Ambroise (2012).
\newblock Variational {Bayesian} inference and complexity control for stochastic block models.
\newblock {\em Statistical Modelling\/}~{\em 12\/}(1), 93--115.
\newblock Publisher: SAGE Publications Sage India: New Delhi, India.

\bibitem[\protect\citeauthoryear{Lee and Wilkinson}{Lee and Wilkinson}{2019}]{lee_review_2019}
Lee, C. and D.~J. Wilkinson (2019, December).
\newblock A review of stochastic block models and extensions for graph clustering.
\newblock {\em Applied Network Science\/}~{\em 4\/}(1), 122.
\newblock Publisher: SpringerOpen.

\bibitem[\protect\citeauthoryear{Legramanti, Rigon, and Durante}{Legramanti et~al.}{2022}]{Le_et_al22-harmonicmeanSBM}
Legramanti, S., T.~Rigon, and D.~Durante (2022).
\newblock Bayesian testing for exogenous partition structures in stochastic block models.
\newblock {\em Sankhya A\/}~{\em 84\/}(1), 108--126.

\bibitem[\protect\citeauthoryear{Li, Liu, Wang, Milkey, Lewis, and Chen}{Li et~al.}{2026}]{Li_et_al25-marglikestimreview}
Li, A., P.-Y. Liu, Y.-B. Wang, A.~Milkey, P.~O. Lewis, and M.-H. Chen (2026).
\newblock A comparison of {M}onte {C}arlo based marginal likelihood estimators.
\newblock {\em Wiley Interdisciplinary Reviews: Computational Statistics\/}~{\em 18\/}(1), e70058.

\bibitem[\protect\citeauthoryear{Lin and Safi}{Lin and Safi}{2024}]{ollamarRpackage}
Lin, H. and T.~Safi (2024, aug).
\newblock ollamar: An r package for running large language models.
\newblock {\em PsyArXiv\/}.

\bibitem[\protect\citeauthoryear{Llorente, Martino, Delgado, and Lopez-Santiago}{Llorente et~al.}{2023}]{llorente_marginal_2023}
Llorente, F., L.~Martino, D.~Delgado, and J.~Lopez-Santiago (2023).
\newblock Marginal likelihood computation for model selection and hypothesis testing: {An} extensive review.
\newblock {\em SIAM Review\/}~{\em 65}, 3--58.

\bibitem[\protect\citeauthoryear{Lomet}{Lomet}{2012}]{lomet_selectiondemodele_2012}
Lomet, A. (2012).
\newblock {\em Sélection de modèle pour la classification croisée de données continues}.
\newblock {PhD} thesis, Université de Technologie de Compiègne.

\bibitem[\protect\citeauthoryear{Mariadassou and Matias}{Mariadassou and Matias}{2015}]{mariadassou_convergence_2015}
Mariadassou, M. and C.~Matias (2015).
\newblock {Convergence of the groups posterior distribution in latent or stochastic block models}.
\newblock {\em Bernoulli\/}~{\em 21\/}(1), 537 -- 573.

\bibitem[\protect\citeauthoryear{Marin, Mengersen, and Robert}{Marin et~al.}{2005}]{Ma_et_al05-allocationRelabelling}
Marin, J.-M., K.~L. Mengersen, and C.~P. Robert (2005).
\newblock Bayesian modelling and inference on mixtures of distributions.
\newblock In D.~K. Dey and C.~R. Rao (Eds.), {\em Handbook of Statistics}, Volume~25, pp.\  459--507. Elsevier.

\bibitem[\protect\citeauthoryear{McDaid, Murphy, Friel, and Hurley}{McDaid et~al.}{2012}]{mcdaid_clustering_2012}
McDaid, A.~F., T.~B. Murphy, N.~Friel, and N.~J. Hurley (2012).
\newblock Clustering in networks with the collapsed stochastic block model.
\newblock {\em arXiv preprint arXiv:1203.3083\/}.

\bibitem[\protect\citeauthoryear{Metodiev}{Metodiev}{2026}]{thamesblock_rpackage}
Metodiev, M. (2026).
\newblock {\em thamesblock: Truncated Harmonic Mean Estimator of the Marginal Likelihood for Block Models}.
\newblock R package version 0.1.0.

\bibitem[\protect\citeauthoryear{Metodiev, Irons, Perrot-Dock{\`e}s, Latouche, and Raftery}{Metodiev et~al.}{025a}]{metodiev2025mixturemodels}
Metodiev, M., N.~J. Irons, M.~Perrot-Dock{\`e}s, P.~Latouche, and A.~E. Raftery (2025a).
\newblock Easily computed marginal likelihoods for multivariate mixture models using the {THAMES} estimator.
\newblock {\em arXiv preprint arXiv:2504.21812\/}.

\bibitem[\protect\citeauthoryear{Metodiev, Perrot-Dockès, Fouetillou, Latouche, and Raftery}{Metodiev et~al.}{2026}]{thamesblock-supp}
Metodiev, M., M.~Perrot-Dockès, G.~Fouetillou, P.~Latouche, and A.~E. Raftery (2026).
\newblock Supplementary material to "simulation-consistent estimation of the marginal likelihood for block models".

\bibitem[\protect\citeauthoryear{Metodiev, Perrot-Dockès, Ouadah, Irons, Latouche, and Raftery}{Metodiev et~al.}{025b}]{metodiev_easily_2025}
Metodiev, M., M.~Perrot-Dockès, S.~Ouadah, N.~J. Irons, P.~Latouche, and A.~E. Raftery (2025b).
\newblock Easily {Computed} {Marginal} {Likelihoods} from {Posterior} {Simulation} {Using} the {THAMES} {Estimator}.
\newblock {\em Bayesian Analysis\/}~{\em 20\/}(3), 1003 -- 1030.
\newblock Publisher: International Society for Bayesian Analysis.

\bibitem[\protect\citeauthoryear{Mouselimis}{Mouselimis}{2024}]{clusterr_rpackage}
Mouselimis, L. (2024).
\newblock {\em {ClusterR}: Gaussian Mixture Models, K-Means, Mini-Batch-Kmeans, K-Medoids and Affinity Propagation Clustering}.
\newblock R package version 1.3.3.

\bibitem[\protect\citeauthoryear{Newman and Reinert}{Newman and Reinert}{2016}]{newman_estimating_2016}
Newman, M. and G.~Reinert (2016, August).
\newblock Estimating the {Number} of {Communities} in a {Network}.
\newblock {\em Physical Review Letters\/}~{\em 117\/}(7), 078301.

\bibitem[\protect\citeauthoryear{Newton and Raftery}{Newton and Raftery}{1994}]{NewtonRaftery1994}
Newton, M.~A. and A.~E. Raftery (1994).
\newblock Approximate {B}ayesian inference with the weighted likelihood bootstrap.
\newblock {\em Journal of the Royal Statistical Society: Series B\/}~{\em 56}, 3--26.

\bibitem[\protect\citeauthoryear{Nobile}{Nobile}{2004}]{No04-MixtureModelsDifferentSizeLink}
Nobile, A. (2004).
\newblock {On the posterior distribution of the number of components in a finite mixture}.
\newblock {\em The Annals of Statistics\/}~{\em 32\/}(5), 2044 -- 2073.

\bibitem[\protect\citeauthoryear{Nobile}{Nobile}{2007}]{No07-MargLikEmptyComponents}
Nobile, A. (2007).
\newblock Bayesian finite mixtures: a note on prior specification and posterior computation.
\newblock {\em arXiv preprint arXiv:0711.0458\/}.

\bibitem[\protect\citeauthoryear{Nobile and Fearnside}{Nobile and Fearnside}{2007}]{nobile_bayesian_2007}
Nobile, A. and A.~T. Fearnside (2007).
\newblock Bayesian finite mixtures with an unknown number of components: {The} allocation sampler.
\newblock {\em Statistics and Computing\/}~{\em 17}, 147--162.
\newblock Publisher: Springer.

\bibitem[\protect\citeauthoryear{Nowicki and Snijders}{Nowicki and Snijders}{2001}]{nowicki_estimation_2001}
Nowicki, K. and T.~A.~B. Snijders (2001).
\newblock Estimation and prediction for stochastic blockstructures.
\newblock {\em Journal of the American statistical association\/}~{\em 96\/}(455), 1077--1087.
\newblock Publisher: Taylor \& Francis.

\bibitem[\protect\citeauthoryear{Ojala and Ripatti-Torniainen}{Ojala and Ripatti-Torniainen}{2024}]{OiRi24-public_review}
Ojala, M. and L.~Ripatti-Torniainen (2024).
\newblock Where is the public of ‘networked publics’? a critical analysis of the theoretical limitations of online publics research.
\newblock {\em European Journal of Communication\/}~{\em 39\/}(2), 145--160.

\bibitem[\protect\citeauthoryear{Papamichalis and Ruane}{Papamichalis and Ruane}{2026}]{Pa_et_al26-collapsed_SBM_variations}
Papamichalis, M. and R.~Ruane (2026).
\newblock Collapsed structured block models for community detection in complex networks.
\newblock {\em arXiv preprint arXiv:2601.02828\/}.

\bibitem[\protect\citeauthoryear{Papastamoulis and Iliopoulos}{Papastamoulis and Iliopoulos}{2010}]{PaIl10-ECR_algo}
Papastamoulis, P. and G.~Iliopoulos (2010).
\newblock An artificial allocations based solution to the label switching problem in {B}ayesian analysis of mixtures of distributions.
\newblock {\em Journal of Computational and Graphical Statistics\/}~{\em 19\/}(2), 313--331.

\bibitem[\protect\citeauthoryear{Peixoto}{Peixoto}{2017}]{peixoto_nonparametric_2017}
Peixoto, T.~P. (2017, January).
\newblock Nonparametric {Bayesian} inference of the microcanonical stochastic block model.
\newblock {\em Phys. Rev. E\/}~{\em 95\/}(1), 012317.
\newblock Publisher: American Physical Society.

\bibitem[\protect\citeauthoryear{Peixoto}{Peixoto}{2019}]{Pe19-SBM_review}
Peixoto, T.~P. (2019, November).
\newblock Bayesian {Stochastic} {Blockmodeling}.
\newblock In P.~Doreian, V.~Batagelj, and A.~Ferligoj (Eds.), {\em Advances in {Network} {Clustering} and {Blockmodeling}\/} (1 ed.)., pp.\  289--332. Wiley.

\bibitem[\protect\citeauthoryear{{R Core Team}}{{R Core Team}}{2025}]{r}
{R Core Team} (2025).
\newblock {\em R: A Language and Environment for Statistical Computing}.
\newblock Vienna, Austria: R Foundation for Statistical Computing.

\bibitem[\protect\citeauthoryear{Robert}{Robert}{2007}]{robert2007bayesian}
Robert, C.~P. (2007).
\newblock {\em The Bayesian choice: from decision-theoretic foundations to computational implementation}.
\newblock Springer.

\bibitem[\protect\citeauthoryear{Rodríguez and Walker}{Rodríguez and Walker}{2014}]{RoWa14-relabellingAlgo}
Rodríguez, C.~E. and S.~G. Walker (2014).
\newblock Label switching in bayesian mixture models: Deterministic relabeling strategies.
\newblock {\em Journal of Computational and Graphical Statistics\/}~{\em 23\/}(1), 25--45.

\bibitem[\protect\citeauthoryear{Stephens}{Stephens}{1997}]{stephens_bayesian_1997}
Stephens, M. (1997).
\newblock {\em Bayesian methods for mixtures of normal distributions}.
\newblock {PhD} thesis, University of Oxford.

\bibitem[\protect\citeauthoryear{Tancini, Rastelli, and Bartolucci}{Tancini et~al.}{2025}]{tancini_marginal_2025}
Tancini, D., R.~Rastelli, and F.~Bartolucci (2025).
\newblock Marginal {Likelihood} for {Intractable} {Hidden} {Markov} {Models}.
\newblock In {\em Scientific {Meeting} of the {Italian} {Statistical} {Society}}, pp.\  398--403. Springer.

\bibitem[\protect\citeauthoryear{Tierney}{Tierney}{1994}]{Ti94-MarkovChainsExploringPosteriors}
Tierney, L. (1994).
\newblock Markov chains for exploring posterior distributions.
\newblock {\em the Annals of Statistics\/}, 1701--1728.

\bibitem[\protect\citeauthoryear{Wang and Wong}{Wang and Wong}{1987}]{wang_stochastic_1987}
Wang, Y.~J. and G.~Y. Wong (1987).
\newblock Stochastic blockmodels for directed graphs.
\newblock {\em Journal of the American Statistical Association\/}~{\em 82\/}(397), 8--19.
\newblock Publisher: Taylor \& Francis.

\bibitem[\protect\citeauthoryear{Wyse and Friel}{Wyse and Friel}{2012}]{WyFr12-collapsed_infinite_LBM}
Wyse, J. and N.~Friel (2012).
\newblock Block clustering with collapsed latent block models.
\newblock {\em Statistics and Computing\/}~{\em 22\/}(2), 415--428.

\bibitem[\protect\citeauthoryear{Yan}{Yan}{2016}]{yan_bayesian_2016}
Yan, X. (2016).
\newblock Bayesian model selection of stochastic block models.
\newblock In {\em Proceedings of the 2016 {IEEE}/{ACM} {International} {Conference} on {Advances} in {Social} {Networks} {Analysis} and {Mining}}, {ASONAM} '16, Davis, California, pp.\  323--328. IEEE Press.

\bibitem[\protect\citeauthoryear{Yang, Zhang, Shen, Ju, and Guo}{Yang et~al.}{2018}]{Ya_et_al18-coreperiphery_twitter}
Yang, J., M.~Zhang, K.~N. Shen, X.~Ju, and X.~Guo (2018).
\newblock Structural correlation between communities and core-periphery structures in social networks: evidence from twitter data.
\newblock {\em Expert Systems with Applications\/}~{\em 111}, 91--99.

\end{thebibliography}

\newcommand{\suppAtext}{Supplement A: proofs}
\newcommand{\textAone}{A.1 Proving that RIS estimators are simulation-consistent and asymptotically normal}
\newcommand{\textAtwo}{A.2 Proving that the ChibPartition estimator is a reciprocal importance sampling estimator}
\newcommand{\textAthree}{A.3 Proving that a HPD region is optimal out of all symmetric truncation sets}
\newcommand{\textAfour}{A.4 Proving an analytic expression of the variance of the THAMES under variational assumptions}
\newcommand{\textAfive}{A.5 Proving that only a part of the permutations is required to compute the THAMES}
\newcommand{\textAsix}{A.6 Proving an extension of Nobile's formula for the SBM}

\newcommand{\suppBtext}{Supplement B: details on the algorithm}
\newcommand{\textBone}{B.1 Choosing the grid for $\alpha$}
\newcommand{\textBtwo}{B.2 Determining the optimal radius of the truncation set}
\newcommand{\textBthree}{B.3 The full algorithm of the THAMES}
\newcommand{\textBfour}{B.4 Details on the computation of the variational approximations of the marginal likelihood}

\newcommand{\suppCtext}{Supplement C: details on the datasets}
\newcommand{\textCone}{C.1 Preprocessing and estimating the marginal likelihood of the COP28 dataset}
\newcommand{\textCtwo}{C.2 Estimating the model parameters of the COP28 dataset}
\newcommand{\textCthree}{C.3 A textual analysis of the COP28 dataset}
\newcommand{\textCfour}{C.4 Comparing the THAMES to the ICL on the COP28 dataset}
\newcommand{\textCfive}{C.5 A sociological interpretation of the results}
\newcommand{\textCsix}{C.6 The specific dataset used in Section 4.3}

\newcommand{\suppDtext}{Supplement D: defining the THAMES for the latent block model}

\newpage
\begin{center}
{\large \textbf{Table of contents}}\\
\end{center}

\begin{flushleft}
{\large\textbf{\suppAtext}} \hfill {\color{blue}\hyperref[sec: suppA]{\pageref{sec: suppA}}} \\
{ \textAone} \hfill {\color{blue}\hyperref[ssec: Aone]{\pageref{ssec: Aone}}} \\
{\textAtwo} \hfill {\color{blue}\hyperref[ssec: Atwo]{\pageref{ssec: Atwo}}} \\
{\textAthree} \hfill {\color{blue}\hyperref[ssec: Athree]{\pageref{ssec: Athree}}} \\
{\textAfour} \hfill {\color{blue}\hyperref[ssec: Afour]{\pageref{ssec: Afour}}} \\
{ \textAfive} \hfill {\color{blue}\hyperref[ssec: Afive]{\pageref{ssec: Afive}}} \\
{ \textAsix} \hfill {\color{blue}\hyperref[ssec: Asix]{\pageref{ssec: Asix}}} \\
\textbf{\suppBtext} \hfill {\color{blue}\hyperref[sec: suppB]{\pageref{sec: suppB}}}\\ 
{\textBone} \hfill {\color{blue}\hyperref[ssec: Bone]{\pageref{ssec: Bone}}} \\
{\textBtwo} \hfill {\color{blue}\hyperref[ssec: Btwo]{\pageref{ssec: Btwo}}} \\
{\textBthree} \hfill {\color{blue}\hyperref[ssec: Bthree]{\pageref{ssec: Bthree}}} \\
{\textBfour} \hfill {\color{blue}\hyperref[ssec: Bfour]{\pageref{ssec: Bfour}}} \\
\textbf{\suppCtext} \hfill {\color{blue}\hyperref[sec: suppC]{\pageref{sec: suppC}}} \\
{\textCone} \hfill {\color{blue}\hyperref[ssec: Cone]{\pageref{ssec: Cone}}} \\
{\textCtwo} \hfill {\color{blue}\hyperref[ssec: Ctwo]{\pageref{ssec: Ctwo}}} \\
{\textCthree} \hfill {\color{blue}\hyperref[ssec: Cthree]{\pageref{ssec: Cthree}}} \\
{\textCfour} \hfill {\color{blue}\hyperref[ssec: Cfour]{\pageref{ssec: Cfour}}} \\
{\textCfive} \hfill {\color{blue}\hyperref[ssec: Cfive]{\pageref{ssec: Cfive}}} \\
{\textCsix} \hfill {\color{blue}\hyperref[ssec: Csix]{\pageref{ssec: Csix}}} \\
\textbf{\suppDtext} \hfill {\color{blue}\hyperref[sec: suppD]{\pageref{sec: suppD}}} \\
\end{flushleft}
\newpage

\section*{\suppAtext \label{sec: suppA}}

Please note that the proof in Supplement {\color{blue}\hyperref[ssec: Aone]{A1}} requires standard conditions on the convergence of the MCMC chain such as ergodicity. Supplement {\color{blue}\hyperref[ssec: Athree]{A3}} requires independence of the MCMC sample and the assumption that there are no two points with exactly equal probabilities (up to permutations), which are much stronger conditions. These conditions can usually only be obtained approximately, for example by thinning the MCMC chain. Supplement {\color{blue}\hyperref[ssec: Afour]{A4}} requires even stronger conditions on the variational assumptions and the concentration of the posterior distribution. However, we would like to emphasise that the simulation-consistency and asymptotic normality of the THAMES (and indeed any RIS estimator) do \textbf{not} depend on these conditions and that the results in Supplements {\color{blue}\hyperref[ssec: Athree]{A3}} and {\color{blue}\hyperref[ssec: Afour]{A4}} are only used to improve the speed of convergence of our estimator, not to ensure convergence itself.

\subsection*{\textAone\label{ssec: Aone}}

A RIS estimator is defined as \begin{align*}
    \hat{Z}^{-1}_{\text{RIS}}(G)=\frac{1}{T}\sum^T_{t=1}\frac{h(C^{(t)})}{p(Y|G,C^{(t)})p(C^{(t)}|G)}.
\end{align*} The goal is to show that this estimator is simulation-consistent, i.e., that $\hat{Z}^{-1}_{\text{RIS}}(G)\stackrel{T\to\infty}\rightarrow \hat{Z}^{-1}(G)$, where convergence is to be understood in the sense of almost sure convergence, and asymptotically normal, in the sense that $\sqrt{T}(\hat{Z}^{-1}_{\text{RIS}}(G)-\hat{Z}^{-1}(G))$ converges in distribution to a normal distribution.

These results apply directly if the MCMC chain is ergodic, as shown by \citet{Ti94-MarkovChainsExploringPosteriors} for a more general case.

\begin{theorem}[Theorem 2.1. in the main article]
Suppose that $h$ is a function such that $\sum_Ch(C)=1.$ If $C^{(1)},\dots,C^{(T)}$ is an ergodic Markov chain with limiting distribution $p(C|G,Y)$, the posterior distribution of $C$ given $Y$ and $G$, then $\hat{Z}^{-1}_{\text{RIS}}(G)$ almost surely converges to $Z^{-1}(G)$. If $C^{(1)},\dots,C^{(T)}$ is additionally ergodic of degree 2, then $\sqrt{T}(\hat{Z}^{-1}_{\textup{RIS}}(G)-Z^{-1}(G))$ converges in distribution to a centred normal distribution.
\end{theorem}

\begin{proof}
Follows directly from \citet[Theorems 3 and 4,][]{Ti94-MarkovChainsExploringPosteriors}.

The only additional conditions for these theorems are that  the expected value of $\frac{h(C)}{p(Y|G,C)p(C|G)}$ is equal to $Z^{-1}(G)$ if $C\sim p(C|Y,G)$ and that $\frac{h(C)}{p(Y|G,C)p(C|G)}$ is bounded with respect to  $C$. The latter condition is trivially true since $C$ can only take finitely many values. The former is also true since\begin{align*}
    E\left[\frac{h(C)}{p(Y|G,C)p(C|G)}\right] = \sum_C \frac{h(C)}{p(Y|G,C)p(C|G)}p(C|G,Y)=\sum_C \frac{h(C)}{Z(G)}=Z^{-1}(G)
\end{align*} due to $\sum_Ch(C)=1$ and Bayes' rule.
\end{proof}

The standard deviation of the THAMES, which depends on the choice of $h$, can for example be estimated by approximating the MCMC sequence via a first-order autoregressive model, as outlined in \citeauthor{metodiev_easily_2025} ({\color{blue}2025b}).

It should be noted that the THAMES estimator for the SBM,\begin{align*}
    \hat{Z}_{\text{THAMES}}^{-1}(G)=\frac{1}{G!}\sum^{G!}_{o=1}\frac{1}{T/2}\sum^T_{\substack{t=T/2+1\\P_o(C^{(t)})\in A}}\frac{1/|A|}{p(C^{(t)}|G)p(Y|G,C^{(t)})},
\end{align*} is also a RIS estimator, even after symmetrising. This can be seen by defining \begin{align*}
    h_{\text{THAMES}}(C)=\frac{1}{G!}\sum^{G!}_{o=1}\frac{\mathds{1}_{P_o(A)}(C)}{|A|},
\end{align*}which is a mixture of uniform distributions on the sets $P_1(A),\dots,P_{G!}(A)$ with all mixture weights set to $\frac{1}{G!}$, and by swapping the order of the sums, such that \begin{align*}
    \hat{Z}_{\text{THAMES}}^{-1}(G)&=\frac{1}{T/2}\sum^T_{t=T/2+1}\frac{\frac{1}{G!}\sum^{G!}_{o=1}\mathds{1}_A(P_o(C^{(t)}))/|A|}{p(C^{(t)}|G)p(Y|G,C^{(t)})}\\&=\frac{1}{T/2}\sum^T_{t=T/2+1}\frac{h_{\text{THAMES}}(C^{(t)})}{p(C^{(t)}|G)p(Y|G,C^{(t)})}.
\end{align*} Hence, the THAMES is a RIS on the second half of the posterior sample.

\subsection*{\textAtwo\label{ssec: Atwo}}

In \citet{hairault_evidence_2022}, the ChibPartition estimator is defined as\begin{align}\label{eq: chibpartition_hairaultetal}
    \hat{Z}_{\text{ChibPartition}}(G)=\frac{p_G(Y|\mathcal{C}^0)\pi_G(\mathcal{C}^0)}{\hat{\pi}_G(\mathcal{C}^0|Y)},
\end{align} where $\mathcal{C}^0$ is a partition of the data points, meaning that $\mathcal{C}^0$ assigns every point to a cluster, ignoring labels of the clusters (e.g.\ the labels (1,1,2,3) and (3,3,2,1) result in the same partition of the data points). Let $\text{partition}(C)$ be a function that maps an allocation vector matrix $C$ to its partition. Then $\pi_G$ is defined as \begin{align}\label{eq: pig_hairaultetal}
    \pi_G(\text{partition}(C))=\frac{G!}{(G-G_+)!}p(C|G),
\end{align} where $G_+$ denotes the number of non-empty clusters implied by $C$, \begin{align}\label{eq: pyf_hairaultetal}
    p_G(Y|\text{partition}(C))=p(Y|G,C),
\end{align} and \begin{align}\label{eq: hatpig_hairaultetal}
    \hat{\pi}_G(\mathcal{C}^0)=\frac{1}{T}\sum^T_{t=1}\mathds{1}_{\{\mathcal{C}^0=\text{partition}(C^{(t)})\}},
\end{align} where $\mathds{1}$ is an indicator function. Also, \begin{align}\label{eq: cmap_hairaultetal}
    \mathcal{C}^0=\text{partition}(C_{\text{MAP}}),
\end{align} where $C_{\text{MAP}}$ is the value of an MCMC sample $C^{(1)},\dots,C^{(T)}$ that maximises\\ $p_G(Y|\text{partition}(C^{(t)}))\pi_G(\text{partition}(C^{(t)}))$.

It turns out that this formulation is equivalent to the formulation given in our article. To show this, let us first note that $C_{\text{MAP}}$ is the same here and in our article (an allocation vector matrix within $C^{(1)},\dots,C^{(T)}$ which maximises the unnormalised posterior probability $p(\hat{C}_{\text{MAP}}|G)p(Y|G,\hat{C}_{\text{MAP}})$), since $p_G(Y|\text{partition}(C^{(t)}))\pi_G(\text{partition}(C^{(t)}))$ is proportional to $p(Y|G,C^{(t)})p(C^{(t)}|G)$. In fact \begin{align}\label{mapequality_hairaultetal}
    p(Y|G,C^{(t)})p(C^{(t)}|G)=p(Y|G,C_{MAP})p(C_{MAP}|G)
\end{align} for any $C^{(t)}$ for which $P_o(C^{(t)})=C_{MAP}$ for some $o\in\{1,\dots,G!\}$, since the collapsed likelihood and prior are invariant with respect to label switching.

It is also the case that\begin{align}\label{eq: indicatorfunc_hairaultetal}
    \mathds{1}_{\{C^{0}=\text{partition}(C^{(t)})\}}=\frac{1}{(G-G_+)!}\sum^{G!}_{o=1}\mathds{1}_{\{C_{MAP}=P_o(C^{(t)})\}}=\frac{1}{(G-G_+)!}\sum^{G!}_{o=1}\mathds{1}_{A_{\text{MAP}}}(C^{(t)}),
\end{align} where $A_{\text{MAP}}$ denotes the set that contains $\hat{C}_{\text{MAP}}$ and all of its label switched versions, since the two sides of the first equality are either both 0 (in the case that there is a permutation $P_o(C^{(t)})$ that is not equal to $C_{MAP}$) or both one, since the existence of one permutation for which $P_o(C^{(t)})=C_{MAP}$ implies the existence of exactly $(G-G_+)!$ such permutations, which are all permutations that can be applied to the empty components of $C^{(t)}$.

Inserting Equation \eqref{eq: pig_hairaultetal}, Equation \eqref{eq: pyf_hairaultetal}, Equation \eqref{eq: cmap_hairaultetal}, Equation \eqref{mapequality_hairaultetal} and Equation \eqref{eq: indicatorfunc_hairaultetal} into Equation \eqref{eq: chibpartition_hairaultetal} gives  \begin{align*}
    \hat{Z}^{-1}_{\text{ChibPartition}}(G)=\frac{1}{T}\sum^T_{\substack{t=1,\\C^{(t)}\in A_{\text{MAP}}}}\frac{1/G!}{p(C^{(t)}|G)p(Y|G,C^{(t)})},
\end{align*} the same expression that was given in the main article.

\subsection*{\textAthree\label{ssec: Athree}}

This proof is an extension of the proof provided in \citeauthor{metodiev2025mixturemodels} ({\color{blue}2025a}). The original proof showed this result for continuous parameters, while we show the result for discrete parameters.

\begin{theorem}[Theorem 3.1. in the main article]\label{thm: optimal_hpd_region}
Let $C^{(T/2+1)},\dots,C^{(T)}\sim p(C|G,Y)$ be an independent sample from the posterior distribution. If no two allocation vector matrices have the same probability (up to permutations), meaning that for any two allocation vector matrices $C^1,C^2$, \begin{align*}
    p(C^1|G,Y)=p(C^2|G,Y)\Rightarrow\exists o\in\{1,\dots,G!\}:P_o(C^1)=C^2,
\end{align*} then there exists an $\alpha\in(0,1]$ such that the variance of $\hat{Z}_{\textup{THAMES}}^{-1}(G)$, conditional on $Y$, is minimised over all symmetric sets by $H_\alpha$. A set $A$ is symmetric if, and only if the image $P_o(A)=A$ for all $o\in\{1,\dots,G!\}$.
\end{theorem}

\begin{proof}
If $A$ is symmetric, then $P_o(C^{(t)})\in A$ if, and only if $C^{(t)}\in A$. Hence \begin{align*}
    \hat{Z}_{\textup{THAMES}}^{-1}(G)&=\frac{1}{G!}\sum^{G!}_{o=1}\frac{1}{T/2}\sum^T_{\substack{t=T/2+1\\P_o(C^{(t)})\in A}}\frac{1/|A|}{p(C^{(t)}|G)p(Y|G,C^{(t)})}
    \\&=\frac{1}{T/2}\sum^T_{\substack{t=T/2+1\\C^{(t)}\in A}}\frac{1/|A|}{p(C^{(t)}|G)p(Y|G,C^{(t)})},
\end{align*} which is proportional to \begin{align*}
    \sum^T_{\substack{t=T/2+1\\C^{(t)}\in A}}\frac{1/|A|}{p(C^{(t)}|G)p(Y|G,C^{(t)})}.
\end{align*} Due to the i.i.d assumption and unbiasedness, this has variance proportional to \begin{align*}
    \text{Var}\left[\left.\sum^T_{\substack{t=T/2+1\\C^{(t)}\in A}}\frac{1/|A|}{p(C^{(t)}|G)p(Y|G,C^{(t)})}\right|G,Y\right]\propto \text{E}\left[\left.\left(\frac{\mathds{1}_{A}(C^{(1)})/|A|}{p(C^{(1)}|G)p(Y|G,C^{(1)})}\right)^2\right|G,Y\right],
\end{align*} where $\mathds{1}$ is an indicator function. The expectation simplifies to \begin{align*}
    \text{E}\left[\left.\left(\frac{\mathds{1}_{A}(C^{(1)})/|A|}{p(C^{(1)}|G)p(Y|G,C^{(1)})}\right)^2\right|G,Y\right]&=\sum_C \left(\frac{\mathds{1}_{A}(C)/|A|}{p(C|G)p(Y|G,C)}\right)^2p(C|G,Y)
    \\&=\frac{1}{p(Y|G)}\sum_C \frac{\mathds{1}_{A}(C)/|A|^2}{p(C|G)p(Y|G,C)},
\end{align*} due to Bayes' rule.

Since there are no posterior ties, any sequence of $C^1,\dots,C^{T_\alpha}$. whose values are ordered such that \begin{align*}
    &p(C^1|Y)=\max_C p(C|G,Y),\\
    &p(C^2|Y)=\max_{C\backslash \{C^1\}}p(C|G,Y),\\
    &\vdots\\
    &p(C^{T_\alpha}|G,Y)=\max_{C\backslash\{C^1,\dots,C^{T_\alpha-1}\}}p(C|G,Y),
\end{align*} defines a HPD region $H_\alpha=\{C^1,\dots,C^{T_\alpha}\}$ for any value $T_\alpha$ where $\alpha$ denotes the probability of this set. Let $T_\alpha=|A|$. Then \begin{align*}
    |H_\alpha|=|A|
\end{align*} and \begin{align*}
    \sum_{C\in H_\alpha}p(C|G,Y)\geq \sum_{C\in A}p(C|G,Y)\Leftrightarrow \sum_{C\in H_\alpha}p(Y|G,C)p(C|G)\geq \sum_{C\in A}p(Y|G,C)p(C|G)
\end{align*} by definition. It follows that \begin{align*}
    \text{Var}[\hat{Z}_{\text{THAMES}}^{-1}(G)&|G,Y]\propto\frac{1}{|A|^2}\sum_{C\in A} \frac{1}{p(C|G)p(Y|G,C)}\\&=\frac{1}{|H_\alpha|^2}\sum_{C\in A\cap H_\alpha} \frac{1}{p(C|G)p(Y|G,C)}\frac{1}{|H_\alpha|^2}\sum_{C\in A\backslash H_\alpha} \frac{1}{p(C|G)p(Y|G,C)}\\&\geq \frac{1}{|H_\alpha|^2}\left(\sum_{C\in A\cap H_\alpha} \frac{1}{p(C|G)p(Y|G,C)}+\frac{1}{q_\alpha}|A\backslash H_\alpha|\right).
\end{align*} This is due to the fact that $A\backslash H_\alpha=\{C\in A|p(C|G)p(Y|G,C)< q_\alpha\}.$ Analogously \begin{align*}
   &\frac{1}{|H_\alpha|^2}\left(\sum_{C\in A\cap H_\alpha} \frac{1}{p(C|G)p(Y|G,C)}+\frac{1}{q_\alpha}|A\backslash H_\alpha|\right)=\\&\frac{1}{|H_\alpha|^2}\left(\sum_{C\in A\cap H_\alpha} \frac{1}{p(C|G)p(Y|G,C)}+\frac{1}{q_\alpha}|H_\alpha\backslash A|\right)\geq\\&\frac{1}{|H_\alpha|^2}\sum_{C\in A\cap H_\alpha}\frac{1}{p(C|G)p(Y|G,C)}+\frac{1}{|H_\alpha|^2}\sum_{H_\alpha\backslash A} \frac{1}{p(C|G)p(Y|G,C)}=\\&\frac{1}{|H_\alpha|^2}\sum_{C\in H_\alpha} \frac{1}{p(C|G)p(Y|G,C)},
\end{align*}where we end up with the same expression, but with $A$ replaced by $H_\alpha$. Thus the variance of the THAMES is decreased or stays the same when choosing $H_\alpha$ instead of $A$. It follows that the variance of the THAMES is minimised by a HPD region.
\end{proof}

\subsection*{\textAfour\label{ssec: Afour}}

In addition to the assumptions in Theorem \ref{thm: optimal_hpd_region}, the following proposition also requires the variational assumptions to hold, and that the relabelled estimator $\hat{z}$ as well as the truncation region $E_{\hat{C},\hat{r}}$ are highly concentrated around one label set. Again, these assumptions rarely hold exactly, but may well hold approximately as the size of the data goes to infinity, since the posterior concentrates around a single point and its label switched versions in this case \citep{celisse_consistency_2012, mariadassou_convergence_2015}.
\begin{proposition}
Let $P_1,\dots,P_{G!}$ denote all possible permutations over any permutation matrix $C$, where $P_1$ denotes the identity. Assume that $C^{(1)},\dots,C^{(T)}$ independently follow the posterior distribution, which is equal to the variational approximation that was randomly permuted due to label switching, \begin{align*}
    p(C|G,Y)=\frac{1}{G!}\sum^{G!}_{o=1}\tilde{h}(P_o(C))=\frac{1}{G!}\sum^{G!}_{o=1}\prod^{n}_{i=1}\textup{Multinom}_G(P_o(C);1,\hat{z}_i),
\end{align*} and that the values of $\hat{z},\hat{C}$ only depend on the values of the first half of the posterior sample. If $E_{\hat{C},\hat{r}},\hat{z}$ are sufficiently concentrated around one set of labels, in the sense that for any allocation vector matrix $C$ \begin{align}\label{eq: concentration_condition_multinom}
    \prod^n_{i=1}\textup{Multinom}_G(P_o(C)_i;1,\hat{z}_i)>0\Rightarrow \prod^n_{i=1}\textup{Multinom}_G(P_{o'}(C)_i;1,\hat{z}_i)=0\quad \forall o'\neq o, 
\end{align} and that \begin{align}\label{eq: concentration_condition_eHatcHatr}
    C\in E_{\hat{C},\hat{r}}\Rightarrow P_o(C)\notin E_{\hat{C},\hat{r}}\quad \forall o\in\{2,\dots,G!\}, 
\end{align} and if the truncation set that defines the THAMES is given by \begin{align}\label{eq: AequaltoE}
    A=E_{\hat{C},\hat{r}}=\{\hat{C}_1^{(1)},\dots,\hat{C}_1^{(\hat{r})}\}\times\dots\times\{\hat{C}_{n}^{(1)},\dots,\hat{C}_n^{(\hat{r})}\},\quad 1\leq \hat{r}\leq T/2,
\end{align} then the variance of the THAMES, conditional on the data $Y$ and the first half of the posterior sample $C^{(1)},\dots,C^{(T/2)}$, is proportional to \begin{align}\label{eq: exactvariance_variational_THAMES}
    \textup{Var}[\hat{Z}_{\textup{THAMES}}^{-1}(G)|G,C^{(1)},\dots,C^{(T/2)},Y]\propto \frac{\prod^n_{i=1}\sum_{C_i\in \{\hat{C}^{(1)}_i,\dots,\hat{C}^{(\hat{r})}_i\}}\frac{1}{\textup{Multinom}_G(C_i;1,\hat{z}_i)}}{\left(\prod^n_{i=1}|\{\hat{C}_i^{(1)},\dots,\hat{C}_i^{(\hat{r})}\}|\right)^2}.
\end{align}
\end{proposition}

\begin{proof}
Following the proof of the last theorem, \begin{align*}
    \text{Var}[\hat{Z}_{\text{THAMES}}^{-1}(G)|G,C^{(1)},\dots,C^{(T/2)},Y] \propto \sum_C\frac{\mathds{1}_A(C)/|A|^2}{p(C|G)p(Y|G,C)}.
\end{align*} Setting $A=E_{\hat{C},\hat{r}}$ and applying Bayes' rule gives \begin{align*}
    \sum_C\frac{\mathds{1}_A(C)/|A|^2}{p(C|G)p(Y|G,C)}\propto \sum_C\frac{\mathds{1}_{E_{\hat{C},\hat{r}}}(C)/|E_{\hat{C},\hat{r}}|^2}{p(C|G,Y)}.
\end{align*} Due to Equation \eqref{eq: concentration_condition_multinom} $p(C|G,Y)$ is either 0 for $C\in\hat{C}$, in which case it is not included in the sum, or there exists exactly one permutation $P_o$ for which $\prod^n_{i=1}\text{Multinom}_G(P_o(C);1,\hat{z}_i)>0$. If $C\in E_{\hat{C},\hat{r}}$, the only such permutation is the identity $C=P_o(C)$ due to Equation \eqref{eq: concentration_condition_eHatcHatr}. It follows that \begin{align}\label{eq: multinom_simplifies_forChat}
    p(C|G,Y)=\prod^n_{i=1}\text{Multinom}_G(C_i;1,\hat{z}_i)
\end{align}for all $C\in E_{\hat{C},\hat{r}}$ with non-zero posterior probability. Thus \begin{align*}
    \sum_C\frac{\mathds{1}_{E_{\hat{C},\hat{r}}}(C)/|E_{\hat{C},\hat{r}}|^2}{p(C|G,Y)}&=\sum_C\frac{\mathds{1}_{E_{\hat{C},\hat{r}}}(C)/|E_{\hat{C},\hat{r}}|^2}{\prod^n_{i=1}\text{Multinom}_G(C_i;1,\hat{z}_i)}
    \\&=\frac{1}{|E_{\hat{C},\hat{r}}|^2}\sum_C\prod^n_{i=1}\frac{\mathds{1}_{E_{\hat{C},\hat{r}}}(C)}{\text{Multinom}_G(C_i;1,\hat{z}_i)}.
\end{align*} Due to the product form of $E_{\hat{C},\hat{r}}$ given in Equation \eqref{eq: AequaltoE}, the sum simplifies to \begin{align*}
    \frac{1}{|E_{\hat{C},\hat{r}}|^2}\sum_C\prod^n_{i=1}\frac{\mathds{1}_{E_{\hat{C},\hat{r}}}(C)}{\text{Multinom}_G(C_i;1,\hat{z}_i)}=\frac{1}{|E_{\hat{C},\hat{r}}|^2}\prod^n_{i=1}\sum_{C_i\in \{\hat{C}^{(1)}_i,\dots,\hat{C}^{(\hat{r})}_i\}}\frac{1}{\text{Multinom}_G(C_i;1,\hat{z}_i)}.
\end{align*} Inserting the squared number of elements of $|E_{\hat{C},\hat{r}}|^2$, \begin{align*}
    |E_{\hat{C},\hat{r}}|^2=\left(\prod^n_{i=1}|\{\hat{C}_i^{(1)},\dots,\hat{C}_i^{(\hat{r})}\}|\right)^2,
\end{align*} into the above expression gives the assertion.
\end{proof}

\subsection*{\textAfive\label{ssec: Afive}}

\begin{theorem}[Theorem 3.1. in the main article]
If there are no empty clusters, in the sense that the event $\sum_{i=1}^n C_{i,g}^{(t)}=0$ does not occur for any value of $T/2+1\leq t\leq T$ and $1\leq g\leq G$, the THAMES is equal to \begin{align}
    \hat{Z}^{-1}_{\textup{THAMES}}(G)=\frac{1}{G!}\sum_{o\in\Omega}\frac{1}{T/2}\sum^T_{\substack{t=T/2+1\\P_o(\psi_\upsilon(C^{(t)}))\in A}}\frac{1/|A|}{p(C^{(t)}|G)p(Y|G,C^{(t)})}.
\end{align}
\end{theorem}

\begin{proof}
As pointed out in \citeauthor{metodiev2025mixturemodels} ({\color{blue}2025a}), sums over all permutations can be exchanged by sums over all orderings, so\begin{align*}
    \hat{Z}^{-1}_{\text{THAMES}}(G)&=\frac{1}{G!}\sum_{o=1}^{G!}\frac{1}{T/2}\sum^T_{\substack{t=T/2+1\\P_o(C^{(t)})\in A}}\frac{1/|A|}{p(C^{(t)}|G)p(Y|G,C^{(t)})}\\&=\frac{1}{G!}\sum_{o=1}^{G!}\frac{1}{T/2}\sum^T_{\substack{t=T/2+1\\P_o(\psi_\upsilon(C^{(t)}))\in A}}\frac{1/|A|}{p(C^{(t)}|G)p(Y|G,C^{(t)})},
\end{align*} since $\psi_{\upsilon}$ just changes the labels of the allocation vector matrices, and we are summing over all permutations anyway. Now $\Omega$ is simply defined as the set that includes all non-zero values of the sum over all of these permutations, which finishes the proof.
\end{proof}

\subsection*{\textAsix\label{ssec: Asix}}

\begin{theorem}[Theorem 3.2. in the main article]

If the prior on the SBM parameters is given by independent draws of\begin{align*}
&\tau|G\sim\textup{Dirichlet}(G,n_1^0,\dots,n_G^0),\quad\tau\in\mathbb{R}^G,\\
&\mu_{g,k}|G\sim\textup{Beta}(a^0_{g,k},b^0_{g,k}),\quad \mu_{g,k}\in(0,1),
\end{align*} then

\begin{align}
    Z(G)=Z(G-1)\cdot\frac{1}{p_0(G)}\cdot\frac{\Gamma(\sum^G_{g=1}n_g^0)\Gamma(n+\sum^{G-1}_{g=1}n_g^0)}{\Gamma(\sum^{G-1}_{g=1}n_g^0)\Gamma(n+\sum^G_{g=1}n_g^0)}.
\end{align}
\end{theorem}

\begin{proof}
While we just use the notation $C$ for an allocation vector matrix in the rest of the document, in this subsection it is useful to distinguish between the allocation vector matrix $C^G$, which is an allocation vector matrix with $n$ rows and $G$ columns, and the allocation vector matrix $C^{G-1}$, which is a matrix with $n$ rows and $G-1$ columns.
 
Let $\tilde{C}^G=\{C^G\in\{0,1\}^{nG}:\sum_{g=1}^GC_{i,g}^G=1\forall i\in\{1,\dots,n\}\}$ denote the space of the allocation vector matrix (left implicit in the main document) for $G\geq 2$. We want to establish a connection between the marginal likelihood of the SBM with $G$ clusters,\begin{align*}
    Z(G)=p(Y|G)=\sum_{C^G\in \tilde{C}^G}p(C|G)p(Y|G,C),
\end{align*} and the marginal likelihood of the SBM with $G-1$ clusters, \begin{align*}
    Z(G-1)=p(Y|G-1)=\sum_{C^{G-1}\in \tilde{C}^{G-1}}p(C^{G-1}|G-1)p(Y|G-1,C^{G-1}).
\end{align*} First of all, let us note that by the law of total probability \begin{align*}
Z(G)&=\sum_{C^G\in \tilde{C}_G}p(C^G|G)p(Y|G,C^G)\\&=\sum_{\substack{C^G\in \tilde{C}^G,\\\sum_{i=1}^nC_{i,G}^G=0}}p(C^G|G)p(Y|G,C^G)+\sum_{\substack{C^G\in \tilde{C}^G,\\\sum_{i=1}^nC_{i,G}^G>0}}p(C^G|G)p(Y|G,C^G).
\end{align*}

Notice that there is a bijection between the space $\{C^G\in\tilde{C}^G:\sum_{i=1}^nC_{i,G}=0\}$ and $\tilde{C}^{G-1}$, since any allocation vector matrix $C^{G-1}$ has exactly one corresponding allocation vector matrix $C^{G}=(C^{G-1},0_n)$ which lives in this space. Indeed, mathematically, there is no difference between the function $p(Y|C^G,G)$ of the SBM with $G$ components, but the $G$'th component being empty, and the function $p(Y|C^{G-1},G-1)$ of the SBM with $G-1$ components. We can show this as follows: let $C^{G-1}$ be an arbitrary allocation vector matrix and $C^{G}=(C^{G-1},0_n)$ the same allocation matrix but with an additional empty component. Then {\footnotesize\begin{align*}
    &p(Y|G,C^{G})=\prod^{G}_{g=1}\prod^G_{k=1}\frac{\Gamma(a_{g,k}^0+b_{g,k}^0)\Gamma(a_{g,k})\Gamma(b_{g,k})}{\Gamma(a_{g,k}+b_{g,k})\Gamma(a_{g,k}^0)\Gamma(b_{g,k}^0)}=
    \\&\prod^{G}_{g=1}\prod^G_{k=1}\frac{\Gamma(a_{g,k}^0+b_{g,k}^0)\Gamma(a_{g,k}^0+\sum^n_{i=1}\sum_{j\neq i}C_{i,g}^GC_{j,k}^GY_{i,j})\Gamma(b_{g,k}^0+\sum^n_{i=1}\sum_{j\neq i}C_{i,g}^GC_{j,k}^G(1-Y_{i,j}))}{\Gamma(a_{g,k}^0+\sum^n_{i=1}\sum_{j\neq i}C_{i,g}^GC_{j,k}^GY_{i,j}+b_{g,k}^0+\sum^n_{i=1}\sum_{j\neq i}C_{i,g}^GC_{j,k}^G(1-Y_{i,j}))\Gamma(a_{g,k}^0)\Gamma(b_{g,k}^0)}=
    \\&\prod^{G-1}_{g=1}\prod^{G-1}_{k=1}\frac{\Gamma(a_{g,k}^0+b_{g,k}^0)\Gamma(a_{g,k}^0+\sum^n_{i=1}\sum_{j\neq i}C_{i,g}^{G-1}C_{j,k}^{G-1}Y_{i,j})\Gamma(b_{g,k}^0+\sum^n_{i=1}\sum_{j\neq i}C_{i,g}^{G-1}C_{j,k}^{G-1}(1-Y_{i,j}))}{\Gamma(a_{g,k}^0+\sum^n_{i=1}\sum_{j\neq i}C_{i,g}^GC_{j,k}^GY_{i,j}+b_{g,k}^0+\sum^n_{i=1}\sum_{j\neq i}C_{i,g}^GC_{j,k}^G(1-Y_{i,j}))\Gamma(a_{g,k}^0)\Gamma(b_{g,k}^0)}\\&=p(Y|G-1,C^{G-1}),
\end{align*}} since \begin{align*}
    &a_{g,k}^0+\sum^n_{i=1}\sum_{j\neq i}C_{i,g}^GC_{j,k}^GY_{i,j}=a_{g,k}^0,
    \\& b_{g,k}^0+\sum^n_{i=1}\sum_{j\neq i}C_{i,g}^GC_{j,k}^G(1-Y_{i,j})=b_{g,k}^0
\end{align*} if either $k=G$ or $j=G$, and in this case the ratio in the above expression reduces to one. However, the prior of the allocation vector matrix changes, since the higher dimension of $C^G$ is penalized: if $C^{G}=(C^{G-1},0_n)$, then \begin{align*}
    \frac{p(C^G|G)}{p(C^{G-1}|G-1)}&=\frac{\frac{\Gamma(\sum^{G}_{g=1}n_g^0)\prod^G_{g=1}\Gamma(n_g^0+\sum^n_{i=1}C_{i,g}^G)}{\Gamma(\sum^{G}_{g=1}(n_g^0+\sum^n_{i=1}C_{i,g}^G))\prod^G_{g=1}\Gamma(n_g^0)}}{\frac{\Gamma(\sum^{G-1}_{g=1}n_g^0)\prod^{G-1}_{g=1}\Gamma(n_g^0+\sum_{i=1}^nC_{i,g}^{G-1})}{\Gamma(\sum^{G-1}_{g=1}(n_g^0+\sum_{i=1}^nC_{i,g}^{G-1}))\prod^{G-1}_{g=1}\Gamma(n_g^0)}}\\&=\frac{\frac{\Gamma(\sum^{G}_{g=1}n_g^0)}{\Gamma(n+\sum^{G}_{g=1}n_g^0)}}{\frac{\Gamma(\sum^{G-1}_{g=1}n_g^0)}{\Gamma(n+\sum^{G-1}_{g=1}n_g^0)}}=\frac{\Gamma(\sum^G_{g=1}n_g^0)\Gamma(n+\sum^{G-1}_{g=1}n_g^0)}{\Gamma(\sum^{G-1}_{g=1}n_g^0)\Gamma(n+\sum^G_{g=1}n_g^0)},
\end{align*} because\begin{align*}
    &n_G=n_G^0+\sum^n_{i=1}C_{i,G}=n_G^0,
    \\&\sum^G_{g=1}n_g^0+\sum^G_{g=1}\sum^n_{i=1}C_{i,g}^G=(\sum^G_{g=1}n_g^0)+\sum^n_{i=1}\sum^G_{g=1}C_{i,g}^G=(\sum^G_{g=1}n_g^0)+n,
    \\&\sum^{G-1}_{g=1}n_g^0+\sum^{G-1}_{g=1}\sum^n_{i=1}C_{i,g}^{G-1}=(\sum^{G-1}_{g=1}n_g^0)+\sum^n_{i=1}\sum^{G-1}_{g=1}C_{i,g}^{G-1}=(\sum^{G-1}_{g=1}n_g^0)+n.
\end{align*} In the case $n_1=\dots=n_G=1$, the ratio even simplifies to \begin{align*}
    \frac{\Gamma(\sum^G_{g=1}n_g^0)\Gamma(n+\sum^{G-1}_{g=1}n_g^0)}{\Gamma(\sum^{G-1}_{g=1}n_g^0)\Gamma(n+\sum^G_{g=1}n_g^0)}=\frac{\Gamma(G)\Gamma(n+G-1)}{\Gamma(G-1)\Gamma(n+G)}=\frac{G-1}{n+G-1},
\end{align*} though this is not relevant for the proof.

We can use these insights to show that \begin{align*}
    &\sum_{\substack{C^G\in \tilde{C}^G,\\\sum_{i=1}^nC_{i,G}^G=0}}p(C^G|G)p(Y|G,C^G)\\&=\frac{\Gamma(\sum^G_{g=1}n_g^0)\Gamma(n+\sum^{G-1}_{g=1}n_g^0)}{\Gamma(\sum^{G-1}_{g=1}n_g^0)\Gamma(n+\sum^G_{g=1}n_g^0)}\sum_{C^{G-1}\in \tilde{C}^{G-1}}p(C^{G-1}|G-1)p(Y|G,C^{G-1})\\&=Z(G-1)\frac{\Gamma(\sum^G_{g=1}n_g^0)\Gamma(n+\sum^{G-1}_{g=1}n_g^0)}{\Gamma(\sum^{G-1}_{g=1}n_g^0)\Gamma(n+\sum^G_{g=1}n_g^0)}.
\end{align*} On the other hand, the other sum can be reformulated using Bayes' rule: \begin{align*}
    &\sum_{\substack{C^G\in \tilde{C}^G,\\\sum_{i=1}^nC_{i,G}^G>0}}p(C^G|G)p(Y|G,C^G)\\&=\sum_{\substack{C^G\in \tilde{C}^G,\\\sum_{i=1}^nC_{i,G}^G>0}}p(C^G|G,Y)p(Y|G)\\&=p(Y|G)(1-\sum_{\substack{C^G\in \tilde{C}^G,\\\sum_{i=1}^nC_{i,G}^G=0}}p(C^G|G,Y))&=p(Y|G)(1-p_0(G))=Z(G)(1-p_0(G)).
\end{align*} It follows that \begin{align*}
    Z(G)=Z(G-1)\frac{\Gamma(\sum^G_{g=1}n_g^0)\Gamma(n+\sum^{G-1}_{g=1}n_g^0)}{\Gamma(\sum^{G-1}_{g=1}n_g^0)\Gamma(n+\sum^G_{g=1}n_g^0)}+Z(G)(1-p_0(G)),
\end{align*} which is equivalent to Equation \eqref{eq: nobileformula_sbm}.

\end{proof}

\section*{\suppBtext \label{sec: suppB}}

\subsection*{\textBone \label{ssec: Bone}}

The truncation set is chosen by selecting $\alpha$ on the grid $\{0,0.1,\dots,1\}$, by computing the sample variance of the THAMES, \begin{align*}
    \frac{1}{T/2-1}\sum^T_{t=T/2+1}\left(\frac{\frac{1}{G!}\sum^{G!}_{o=1}\mathds{1}_A(P_o(C^{(t)}))/|A|}{p(Y|G,C^{(t)})p(C^{(t)}|G)}-\hat{Z}^{-1}_{\text{THAMES}}\right)^2
\end{align*} and by adjusting it for autoregression, using the function ar from the programming language R \citep{r}. More precisely, a first-order autoregressive model with parameter
$\phi$ is used, and the spectral density of the series at zero is approximated as $1/(1-\phi)^2$. The variance is multiplied by this estimator. $\alpha$ is then chosen such that it minimises the obtained variance estimate over the grid.

\subsection*{\textBtwo \label{ssec: Btwo}}

$\hat{r}$ can be chosen by computing Equation \eqref{eq: exactvariance_variational_THAMES} and by selecting the value of $\hat{r}$ that minimises it. However, it is important that $E_{\hat{C},\hat{r}}$ is not too small, since the THAMES could possibly not be defined in this case (since the sum may evaluate to 0). This is why we choose $\hat{r}$ such that it minimises the above variance under the condition that the probability of $E_{\hat{C},\hat{r}}$ being empty is small, i.e., under the condition that \begin{align*}
    \sum_{C\in E_{\hat{C},\hat{r}}}\prod^n_{i=1}\text{Multinom}_G(C_i;1,\hat{z}_i)\geq 20/T,
\end{align*} since in this case the expected number of non-empty elements in the series is at least equal to 20, under the above assumptions, due to Equation \eqref{eq: multinom_simplifies_forChat}. If there exists no value of $\hat{r}$ for which this condition is fulfilled, we set $\hat{r}=T/2$.

\subsection*{\textBthree\label{ssec: Bthree}}
The THAMES is computed using Algorithm \ref{algorithm: thamesblockrecursion}, which is recursive: in the case of empty clusters, it computes a lower dimensional MCMC sample and calls itself, using this MCMC sample. If there are no empty clusters, it computes the THAMES through Algorithm \ref{algorithm: thamesblockalgo}.

The parameter $\alpha$ in Algorithm \ref{algorithm: thamesblockalgo} is optimised through a for loop: a specific grid is chosen (in practice we choose $\{0,0.1,0.2,0.3,0.4,0.5,0.6,0.7,0.8,0.9,1\}$) and the THAMES as well as its sample variance is computed. Finally, the value of the THAMES with the smallest sample variance is chosen. If $\alpha$ leads to an empty truncation set, the THAMES is not defined. In this case, we define the sample variance to be equal to infinity.

{ \small \begin{algorithm}[H]
	\KwIn{MCMC sample $(C^{(t)})_{t\in\llbracket 1,T\rrbracket}$, unnormalised posterior $p(C)\cdot p(Y|C)$}
	
	\textbf{Cartesian product:} $E_{\hat{C},r}=\{\hat{C}_1^{(1)},\dots,\hat{C}_1^{(r)}\}\times\dots\times\{\hat{C}_{n}^{(1)},\dots,\hat{C}_n^{(r)}\},\quad 1\leq r\leq T/2.$
	
	\textbf{Sample splitting:} set\begin{align*}
	    \hat{r}=\text{argmin}_{r}\frac{\prod^n_{i=1}\sum_{C_i\in \{\hat{C}^{(1)}_i,\dots,\hat{C}^{(r)}_i\}}\frac{1}{\textup{Multinom}_G(C_i;1,\hat{z}_i)}}{\left(\prod^n_{i=1}|\{\hat{C}_i^{(1)},\dots,\hat{C}_i^{(r)}\}|\right)^2},
	\end{align*} under the condition that \begin{align*}
    \sum_{C\in E_{\hat{C},\hat{r}}}\prod^n_{i=1}\text{Multinom}_G(C_i;1,\hat{z}_i)\geq 20/T,
\end{align*} where $(\hat{C}^{(1)},\dots,\hat{C}^{(T/2)})$ are ordered such that $p(\hat{C}^{(1)})p(Y|\hat{C}^{(1)})\geq \cdots\geq p(\hat{C}^{(T/2)})p(Y|\hat{C}^{(T/2)})$, and \begin{align*}
    \hat{z}_{i,g}=\frac{1+\sum_{t=1}^{T/2}C_{i,g}^{(t)\star}}{2+T/2}.
\end{align*}

\textbf{Uniform sample:} sample $\eta^{(1)},\dots,\eta^{(T)}\sim\text{Unif}(E_{\hat{C},\hat{r}})$.

\For{$\alpha\in\{0,\dots, 1\}$}{
\textbf{Truncation subset:} $A=B_{\hat{C},\hat{r},\alpha}=E_{\hat{C},\hat{r}}\cap\{C:p(Y|G,C)p(C|G)>\hat{q}_\alpha\}$ where $\hat{q}_\alpha$ is the $1-\alpha$ quantile of $(p(C^{(t)})p(Y|C^{(t)}))_{t\in\llbracket 1,T\rrbracket}$.

\textbf{Size estimation:} \begin{align*}
    |A|\simeq \widewedge{\lvert B_{\hat C,\hat r,\alpha}\rvert}=\frac{1}{T}\sum_{\substack{t=1\\\eta^{(t)}\in H_\alpha}}^T|E_{\hat{C},\hat{r}}|.
\end{align*}

\textbf{Lexicographical ordering} choose $\upsilon_g=\text{argmax}_i \hat{z}_{i,g}$ and $\psi_\nu$ such that \begin{align*}
    \min\{i\in\upsilon:\psi_\upsilon(C)_{i,1}=1\}>\dots> \min\{i\in\upsilon:\psi_\upsilon(C)_{i,G}=1\},
\end{align*} for any value of $C$.

\textbf{Permutation subset}: $o\in\Omega$ if $P_o$ does not permute the cluster of any row $i$ of the second half of the posterior sample for which \begin{align*}
    |\{\hat{C}_i^{(1)},\dots,\hat{C}_i^{(\hat{r})}\}|=1.
\end{align*}

\textbf{Calculate} \begin{align*}
    \hat{Z}^{-1}_{\text{THAMES}}(G)=\frac{1}{G!}\sum_{o\in\Omega}\frac{1}{T/2}\sum^T_{\substack{t=T/2+1\\P_o(\psi_\upsilon(C^{(t)}))\in A}}\frac{1/|A|}{p(C^{(t)}|G)p(Y|G,C^{(t)})},
\end{align*} with $|A|$ replaced by its estimate.

\textbf{Calculate} \begin{align*}
   \hat{\text{V}}[\hat{Z}^{-1}_{\text{THAMES}}(G)]= \frac{1/(1-\phi)^2}{T/2-1}\sum^T_{t=T/2+1}\left(\frac{\frac{1}{G!}\sum_{o\in\Omega}\mathds{1}_A(P_o(C^{(t)}))/|A|}{p(Y|G,C^{(t)})p(C^{(t)}|G)}-\hat{Z}^{-1}_{\text{THAMES}}\right)^2,
\end{align*} where $1/(1-\phi)^2$ approximates the spectral density of the series at zero, and $|A|$ being replaced by its estimate.

}

\textbf{Output:} the estimate $\hat{Z}^{-1}_{\text{THAMES}}(G)$ of the reciprocal marginal likelihood with $\alpha$ chosen such that it minimises $\hat{\text{V}}[\hat{Z}^{-1}_{\text{THAMES}}(G)]$.
	\caption{THAMES for block models}
	\label{algorithm: thamesblockalgo}
	\end{algorithm}}
	
{ \small \begin{algorithm}[H]
	\KwIn{MCMC sample $(C^{(t)})_{t\in\llbracket 1,T\rrbracket}$, hyperparameters $n_1^0,\dots,n_{G-1}^0$, unnormalised posterior $p(C)\cdot p(Y|C)$}

\If{$\exists g\in\{1,\dots,G\},t\in\{1,\dots,T\}:\sum_{i=1}^n C_{i,g}^{(t)}=0$ and $G>2$}{
\textbf{Generate} a new MCMC sample with respect to $G-1$.

\textbf{Calculate} $\hat{Z}^{-1}_{\text{THAMES}}(G-1)$ through Algorithm \ref{algorithm: thamesblockrecursion}.

\textbf{Calculate} \begin{align*}\hat{p}_0(G)=\frac{1}{G}\sum^G_{g=1}\frac{1}{T}\sum_{t=1}^T\prod_{i=1}^n (1-C_{i,g}^{(t)}).
\end{align*}

\textbf{Calculate}\begin{equation*}
    \hat{Z}^{-1}_{\text{THAMES}}(G)=\left(\frac{\Gamma\left(\sum^G_{g=1}n_g^0\right)\Gamma\left(n+\sum^{G-1}_{g=1}n_g^0\right)}{\Gamma\left(n+\sum^G_{g=1}n_g^0\right)\Gamma\left(\sum^{G-1}_{g=1}n_g^0\right)}\cdot \hat{Z}^{-1}_{\text{THAMES}}(G-1)\cdot\frac{1}{\hat{p}_0(G)}\right)^{-1}.
\end{equation*}
}
\Else{\textbf{Calculate} $\hat{Z}^{-1}_{\text{THAMES}}(G)$ through Algorithm \ref{algorithm: thamesblockalgo}.}

\textbf{Output:} the estimate $\hat{Z}^{-1}_{\text{THAMES}}(G)$ of the reciprocal marginal likelihood \caption{Recursive THAMES for block models}
	\label{algorithm: thamesblockrecursion}
	\end{algorithm}}

\subsection*{\textBfour\label{ssec: Bfour}}

The ILvb \citep{latouche_variational_2012} is defined as \begin{align*}
    \text{IL}_{vb}&=\text{ln}\{\frac{\Gamma(\sum^G_{g=1}n_g^0)\prod^G_{g=1}\Gamma(n_g)}{\Gamma(\sum^G_{g=1}n_g)\prod^G_{g=1}\Gamma(n_g^0)}\}
    \\&+\sum^G_{g\leq k}\ln\{\frac{\Gamma(a^0_{g,k}b^0_{g,k})\Gamma(a_{g,k})\Gamma(b_{g,k})}{\Gamma(a_{g,k}+b_{g,k})\Gamma(a^0_{g,k})\Gamma(b^0_{g,k})}\}-\sum^n_{i=1}\sum^G_{g=1}\tilde{z}_{i,g}\text{ln}\tilde{z}_{i,g},
\end{align*} where $n_g=n_g^0+\sum^n_{i=1}C_{i,g}$, $a_{g,k}=a^0_{g,k}+\sum^n_{i=1}\sum_{j\neq i}C_{i,g}C_{j,k}Y_{i,j}$, and $ b_{g,k}=b^0_{g,k}+\sum^n_{i=1}\sum_{j\neq i}C_{i,g}C_{j,k}(1-Y_{i,j})$, as in the main document, and where $\tilde{z}$ are cluster probabilities estimated by a variational algorithm such that they maximise the above expression. We implemented this algorithm \citep[using the $"$kmeans$"$ function from the programming language R with parameter nstart=10,][for initisalisation]{r}. To make sure that it was robust with respect to the starting value, it was independently run 10 times, selecting the value of ILvb that was the highest.

The correction of the ILvb, ILvbc, is simply defined as \begin{align*}
    \text{IL}_{vbc}=\text{ln}(G!) + \text{IL}_{vb}.
\end{align*}

\section*{\suppCtext \label{sec: suppC}}

\subsection*{\textCone \label{ssec: Cone}}

\begin{table}[]
    \centering
    \begin{tabular}{cccccc}
                  G & 2 & 3 & 4 & 5 & 6\\\hline 
                  THAMES & -543854 & -520356 & -505344 & -495575 & -487410\\\hline
                  G & 7 & 8 & 9 & 10 & 11\\\hline
                  THAMES & -482603 & -477779 & -469635 & -467160 & -462584\\\hline
                  G & \textbf{12} & 13 & 14 & 15 & 16\\\hline
                  THAMES & \textbf{-460182} & -460189 & -460195 & -460202 & -460208\\\hline
                  G & 17 & 18 & 19 & 20 &\\\hline
                  THAMES & -460215 & -460221 & -460227 & -460234 & 
    \end{tabular}
    \caption{THAMES estimates of the log marginal likelihood; the largest estimate is obtained at $G=12$; all estimates were rounded to full integers.}
    \label{tab: cop28_thames_marglikestims}
\end{table}

The COP28 dataset is too large to allow fast computation with our available processing power. We therefore reduced it by including only users that were quoted and/or reposted at least 15 times and by selecting the largest connected component of the resulting graph, which contains 10,979 nodes and 74,826 edges. All of the following analysis was performed on the resulting dataset.

Since the dataset is too large to use the standard kmeans function, the function MiniBatchKmeans \citep{minibatchkmeans_rpackage} from the ClusterR package \citep{clusterr_rpackage} was used for initialisation, with a batch size of 1,000 and a tolerance value of 1e-6. This initialisation was used to run 10 independent Gibbs samplers for each value of $G\in\{2,\dots,20\}$. Since some of the Gibbs samplers got stuck in local modes with a smaller number of clusters than the initial value, only the sample with the highest MAP value of $C^{(1)},\dots,C^{(T)}$ was selected, with the remaining samples being discarded. After a burn-in of 2,000, a sample of length $T=10,000$ was produced for each value of $G$. The THAMES was computed on each of these samples.

The marginal likelihood estimates obtained by the THAMES for the COP28 for $G\in\{2,\dots,20\}$ are shown in Table \ref{tab: cop28_thames_marglikestims}. It is estimated that the marginal likelihood is maximised by $G=12$. Interestingly, the estimates are strictly increasing until this value is obtained, and strictly decreasing after. Now that $G=12$ is established as an estimate of the number of clusters, the other model parameters of the SBM can also be estimated.

\subsection*{\textCtwo \label{ssec: Ctwo}}

\begin{table}[]
    \centering
    \begin{tabular}{ccccccc}
        cluster&1&2&3&4&5&6\\\hline
        size&73&234&68&423&6222&1172\\
        cluster&7&8&9&10&\textbf{11}&12\\\hline
        size&1377&44&717&278&\textbf{5}&366
    \end{tabular}
    \caption{estimated cluster sizes of the COP28 dataset; the size of cluster 11 is highlighted.}
    \label{tab: cop28_estimated_cluster_sizes}
\end{table} 

For $G=12$, a MAP estimator $\hat{C}^{MAP}$ was obtained by maximising the unnormalised posterior density $p(C^{(t)}|G)p(Y|G,C^{(t)})$ for $t=1,\dots,T$. The clusters obtained by $\hat{C}^{MAP}$ are the ones shown in the main article. Their sizes are shown in Table \ref{tab: cop28_estimated_cluster_sizes}. They are quite variable, with cluster 5, a cluster of size 6222, being the largest cluster, and cluster 11, a cluster of size 5, being the smallest cluster.

$\hat{C}^{MAP}$ can be used to obtain estimates of the SBM parameters $\mu$ and $\tau$ by maximising the \textbf{complete-data loglikelihood} \begin{align*}
        \log(p(Y,\hat{C}^{MAP}|G,\tau,\mu))&=\sum_{i=1}^n\sum_{g=1}^G \hat{C}^{MAP}_{i,g}\log(\tau_{g})\\&+\sum_{i=1}^n\sum_{g=1}^G\sum_{j\neq i}\sum_{k=1}^G\hat{C}_{i,g}^{MAP}C_{j,k}^{MAP}(Y_{i,j}\log(\mu_{g,k})+(1-Y_{i,j})\log(1-\mu_{g,k})).
    \end{align*} 
    
    \begin{table}[]
    \centering
    \begin{tabular}{cccccccccccc}
        &&&&&&{\huge $\hat{\tau}$}&&&&&\\\hline
        $\hat{\tau}_1$ & $\hat{\tau}_2$ & $\hat{\tau}_3$ & $\hat{\tau}_4$ & $\hat{\tau}_5$ & $\hat{\tau}_6$ & $\hat{\tau}_7$ & $\hat{\tau}_8$ & $\hat{\tau}_9$ & $\hat{\tau}_{10}$ & $\hat{\tau}_{11}$ & $\hat{\tau}_{12}$\\\hline
        0.66 & 2.13 & 0.62 & 3.85 & 56.67 & 10.67 & 12.54 & 0.40 & 6.53 & 2.53 & \textbf{0.05} & 3.33 
    \end{tabular}
    
    \begin{tabular}{cccccccccccc}
    &&&&&&{\huge$\hat{\mu}$}&&&&&\\\hline
    $\hat{\mu}_1$ & $\hat{\mu}_2$ & $\hat{\mu}_3$ & $\hat{\mu}_4$ & $\hat{\mu}_5$ & $\hat{\mu}_6$ & $\hat{\mu}_7$ & $\hat{\mu}_8$ & $\hat{\mu}_9$ & $\hat{\mu}_{10}$ & $\hat{\mu}_{11}$ & $\hat{\mu}_{12}$\\\hline
                14.28 & 0.07 & 2.14 & 9.99 & 0.02 & 0.02 & 0.01 & 55.45 &0.01 & 0.00 & \textbf{28.77} & 0.51\\\hline
                0.00 & 1.89 & 6.98 & 0.00 & 0.02 & 0.02 & 0.23 & 0.16 & 0.01 & 0.33 & \textbf{17.86} & 0.17\\\hline
                0.00 & 4.80 & 16.89 & 0.01 & 0.04 & 0.04 & 0.64 & 0.37 & 0.01 & 1.02 & \textbf{36.18} & 0.68\\\hline
                0.40 & 0.00 & 0.20 & 0.31 & 0.00 & 0.00 & 0.00 & 7.58 & 0.00 & 0.00 & \textbf{6.57} & 0.05\\\hline
                0.00 & 0.02 & 0.09 & 0.00 & 0.01 & 0.01 & 0.00 & 0.02 & 0.01 & 0.00 & \textbf{1.19} & 0.09\\\hline
                0.00 & 0.20 & 1.27 & 0.01 & 0.07 & 0.19 & 0.03 & 0.28 & 0.02 & 0.05 & \textbf{12.80} & 0.94\\\hline
                0.00 & 0.73 & 3.12 & 0.00 & 0.01 & 0.01 & 0.16 & 0.13 & 0.00 & 0.17 & \textbf{11.02} & 0.13\\\hline
                0.31 & 0.10 & 0.37 & 0.41 & 0.02 & 0.05 & 0.06 & 6.97 & 0.02 & 0.05 & \textbf{8.64} & 0.27\\\hline
                0.00 & 0.01 & 0.07 & 0.00 & 0.02 & 0.00 & 0.00 & 0.06 & 0.94 & 0.01 & \textbf{1.76} & 0.05\\\hline
                0.00 & 6.41 & 22.46 & 0.03 & 0.11 & 0.13 & 0.87 & 1.05 & 0.05 & 1.41 & \textbf{51.01} & 1.30\\\hline
                \textbf{0.00} & \textbf{0.09} & \textbf{0.88} & \textbf{0.00} & \textbf{0.00} & \textbf{0.02} & \textbf{0.00} & \textbf{0.00} & \textbf{0.00} & \textbf{0.00} & \textbf{8.00} & \textbf{0.05}\\\hline
                0.00 & 0.10 & 0.72 & 0.00 & 0.02 & 0.06 & 0.02 & 0.19 & 0.00 & 0.02 & \textbf{6.99} & 0.40
    \end{tabular}
    
    \caption{estimates of the SBM parameters (in percent); parameters corresponding to cluster 11 are highlighted. In the context of the COP28 dataset, $\hat{\mu}_{g,k}$ estimates the probability that a user in cluster $g$ quotes or reposts a user in cluster $k$ while $\hat{\tau}_g$ estimates the probability that a node chosen at random belongs to cluster $g$. Estimates are rounded to the second decimal place.}
    \label{tab: cop28_sbm_parameter_estimates}
\end{table}  

\begin{figure}
    \centering
    \includegraphics[scale=.8]{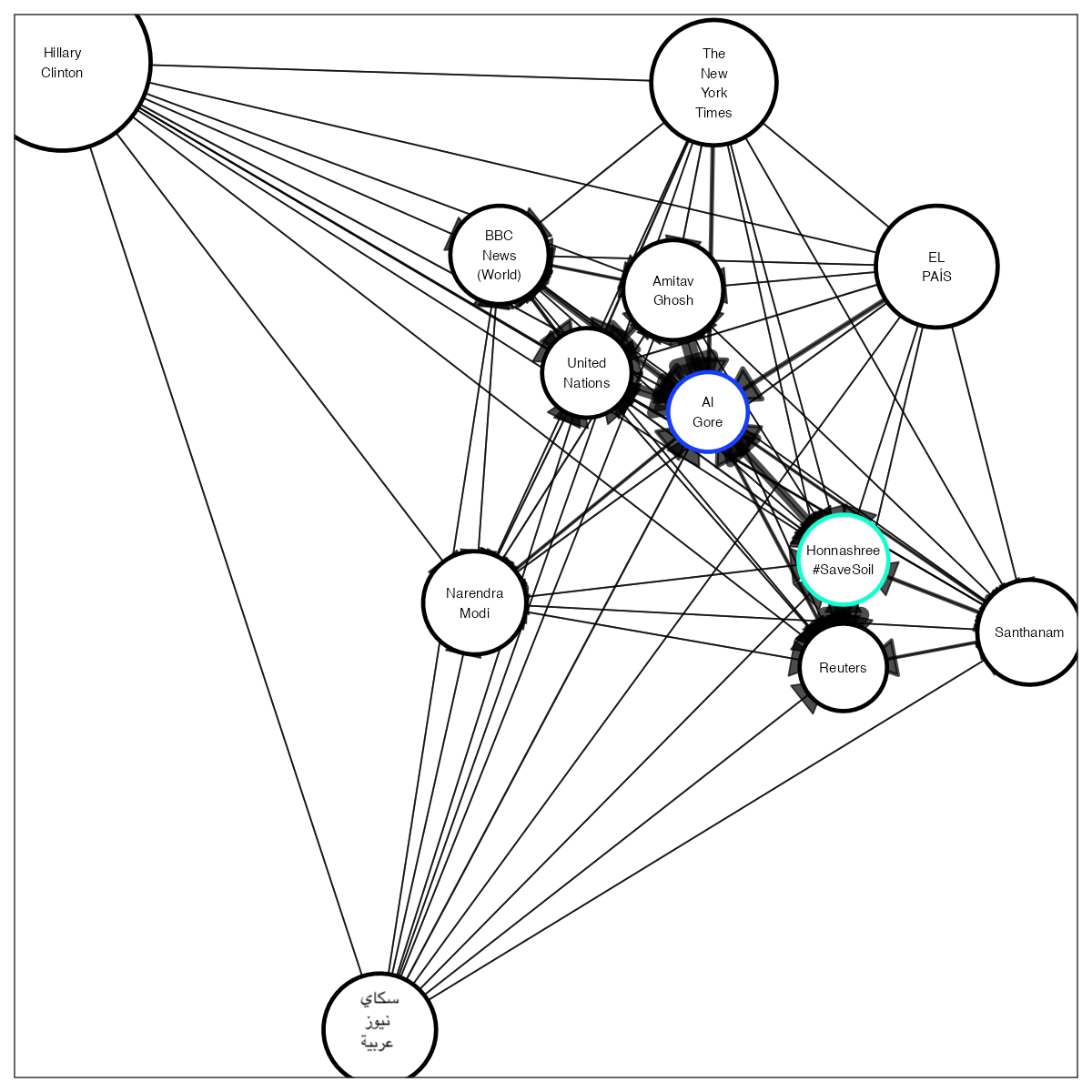}
    \caption{the metagraph of the COP28 dataset for $G=12$ clusters, the number that maximises the marginal likelihood estimates of the THAMES; Clusters are represented by nodes. Large edges indicate large connection probabilities between clusters, while large nodes indicate a large relative size of the cluster. $"$Al Gore$"$ represents the core of the network, while $"$Honnashree \# SaveSoil$"$ represents the cluster closest to the sociological concept of a community, formed around the $\#\text{SaveSoil}$ movement.}
    \label{fig: thames_metagraph}
\end{figure}
    
    As shown in \citet{daudin_mixture_2008} this function is maximised by\begin{align*}
        \hat{\mu}_{g,k}=\frac{\sum_{i=1}^n\sum_{j\neq i}^n\hat{C}_{i,g}^{MAP}\hat{C}_{j,k}^{MAP}Y_{i,j}}{\sum_{i=1}^n\sum_{j\neq i}^n\hat{C}_{i,g}^{MAP}\hat{C}_{j,k}^{MAP}},
    \end{align*} and\begin{align*}
        \hat{\tau}_{g}=\frac{1}{n}\sum_{i=1}^n\hat{C}_{i,g}^{MAP}.
    \end{align*} These estimates are shown in Table \ref{tab: cop28_sbm_parameter_estimates}. It is clear that any user from any cluster different to cluster 11 has a high probability of quoting/reposting a user from cluster 11, while there is a very low probability that a user from cluster 11 quotes or reposts a user outside of cluster 11. This is true even though cluster 11 is the smallest out of all other clusters. Thus, the clustering found can be identified as a core-periphery structure, where all clusters but cluster 11 are at the periphery, and cluster 11 is its core.
    
    A particularly elegant way to summarise these results is given by a metagraph. In a metagraph, the clusters are represented by nodes. The sizes of the nodes represent the sizes of the clusters, while the sizes of the edges between clusters represent the conncection probabilities. Such a metagraph is shown in Figure \ref{fig: thames_metagraph}. The labels of the nodes show the username of the user with the highest number of followers (taken as the maximum of this number over all timepoints on which the user sent a post within this dataset). $"$Al Gore$"$ represents cluster 11, whose central function as the core of the dataset is quite clear, as all other clusters are quoting or reposting its users with fairly high probabilities.
    
    Interestingly, while there aren't any clusters which correspond to the sociological notion of a community in the strict sense of the word — as outlined in the main article — cluster 1, represented by the user $"$Honnashree \# SaveSoil$"$ may be quite close to this definition, as its members do quote/repost each other fairly often and as it is centred around the same topic, the $\#\text{SaveSoil}$ movement. This can be seen when conducting a textual analysis.

\subsection*{\textCthree \label{ssec: Cthree}}

\paragraph{\textbf{Preprocessing}}

A text corpus of the COP28 dataset is available as each of the 4,373,354 posts concerning the COP28 from the 1st of January 2023 to the 31st of December 2023 contains textual data. Of these posts, only 179,885 are reposted/quoted by people within the reduced network of 10,979 nodes and 74,826 edges used in the main article. The texts within these posts were preprocessed via the following steps: \begin{itemize}
    \item punctuation signs, links, and the user name of the post were removed.

    \item common stopwords were removed via the R package tm \citep{tmpackage}.
    
    \item words with less than 5 characters were removed.

    \item stemming based on the English language was performed via the R package SnowballC \citep{packageSnowballC}.
\end{itemize}

\paragraph{\textbf{Results}}

A textual analysis was performed via the Latent dirichlet alocation (LDA) model \citep{Bl_et_al03-LDAintroduction}. This model is characterised by the idea that there is a hidden probability vector through which a fixed number of topics is drawn, and each word is drawn through a different distribution given its respective topic. This is particularly interesting to us because it allows us to estimate a variational estimate of topic probabilities $\hat{\gamma}$, where each row of $\hat{\gamma}$ is a probability vector corresponding to one specific document (post) in the corpus \citep{GrKu11-topicmodels_rpackage}. $\hat{\gamma}$ is used to analyse the topics of each cluster pair by averaging the rows of $\hat{\gamma}$ that correspond to posts being made with respect to this cluster pair. For example, if a user from cluster 1 quoted or reposted 10 posts from a user of cluster 2, the topic distribution for this exchange can be obtained by summing the rows from $\hat{\gamma}$ that correspond to all words included in these posts.

$\hat{\gamma}$ was fitted by using the function LDA from the topicmodels R package \citep{topicmodelspackage}. 40 topics were fitted for the whole corpus. These topics are summarised in Table \ref{tab: lda_all_topics}. 

A topic-specific summary was obtained as follows: first, an estimate $\hat{\beta}$ is computed via the topicmodels package, where each entry $\hat{\beta}_{\tilde{i},\tilde{j}}$ is a matrix that gives the probability that a word with index $\tilde{i}$ appears given the topic with index $\tilde{j}$ \citep{GrKu11-topicmodels_rpackage}.

Then, the entropy \begin{align*}
    \hat{\beta}_{\tilde{i},\tilde{j}}\sum_{i'} \hat{\beta}_{i',\tilde{j}}\log(\hat{\beta}_{i',\tilde{j}})/\log(\tilde{K})+\beta_{\tilde{i},\tilde{j}}
\end{align*} is computed for each word $\tilde{i}$ and each topic $\tilde{j}$. For every topic $\tilde{j}$ those 200 words are chosen that maximise the entropy. Those 200 words where then summarised by the large language model gamma 4, using the R package ollamar \citep{ollamarRpackage}. Notably, almost all of the topics have very general summaries related to climate change, with topic 5 being one of the few that is more precise. It seems related to the $\#\text{SaveSoil}$ movement, which, according to its {\color{blue}\href{https://consciousplanet.org/en/save-soil}{website}}, is $"$a global movement to address the soil crisis by uniting people across the globe to stand up for soil health, and support leaders of all nations in actioning policies toward increasing organic matter in agricultural soil$"$. Interestingly, this topic is also related to a specific cluster located by our algorithm. This can be seen by estimating the topic distributions of each cluster pair.

\begin{table}[]
    \centering
     \begin{tabular}{cc}
        topic & summary \\\hline
 1&Protect planet's natural ecosystems\\
2&Diverse global vocabulary collection\\
3&Healthy soils combat climate change\\
4&Global climate action efforts\\
\textbf{5}&\textbf{Soil health, climate, global action}\\
6&Addressing climate change globally\\
7&Urgent global climate response needed\\
8&Climate action, healthy soil, future\\
9&Climate solutions require global cooperation\\
10&Global climate action needed\\
11&Global energy transition demands sustainability\\
12&No coherent text provided\\
13&Global climate health action\\
14&Global crises: conflict, climate, human rights\\
15&Climate action and energy transition\\
16&Climate change mitigation strategies\\
17&Global climate action agreements\\
18&Protect soil, stabilize climate\\
19&Sustainable energy for climate action\\
20&Climate action needed globally\\
21&Climate action confronts big polluters\\
22&Healthy soil combats climate change\\
23&Global investment for climate transition\\
24&Global climate politics and action\\
25&Climate justice, human rights, peace\\
26&Climate politics face global crises\\
27&Climate action and sustainable growth\\
28&Sustainable food systems, climate action\\
29&Climate action, gender equality\\
30&Climate action requires global cooperation\\
31&Global efforts curb warming emissions\\
32&Climate justice needs global action\\
33&Global climate action in UAE\\
34&Climate action for global health\\
35&Decarbonization requires systemic shift\\
36&Climate change demands policy action\\
37&Climate crisis demands global action\\
38&Climate action for food security\\
39&Global climate action initiatives\\
40&Climate change needs global action
    \end{tabular}
    \caption{all of the 40 topics obtained by fitting LDA to the COP28 dataset; each topic is summarised by selecting the 200 words with the lowest entropy with respect to this topic, and then by summarising these words via gemma 4. Topic 5, which relates to the $\#\text{SaveSoil}$ movement, is emphasised.}
    \label{tab: lda_all_topics}
\end{table}
\clearpage
\thispagestyle{empty}
\begin{figure}
    \centering
    \begin{tabular}{cc}
    \begin{tabular}{c}
         \includegraphics[scale=1]{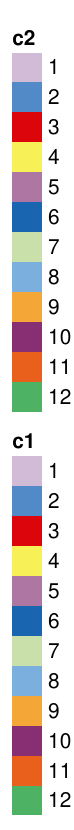}\\\includegraphics[scale=1]{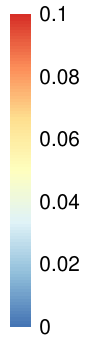}\vspace{20cm}
    \end{tabular}&\vspace{-18cm}\includegraphics[angle=90,height=1\textheight,trim=0cm 1.2cm 0cm 0cm,
    clip]{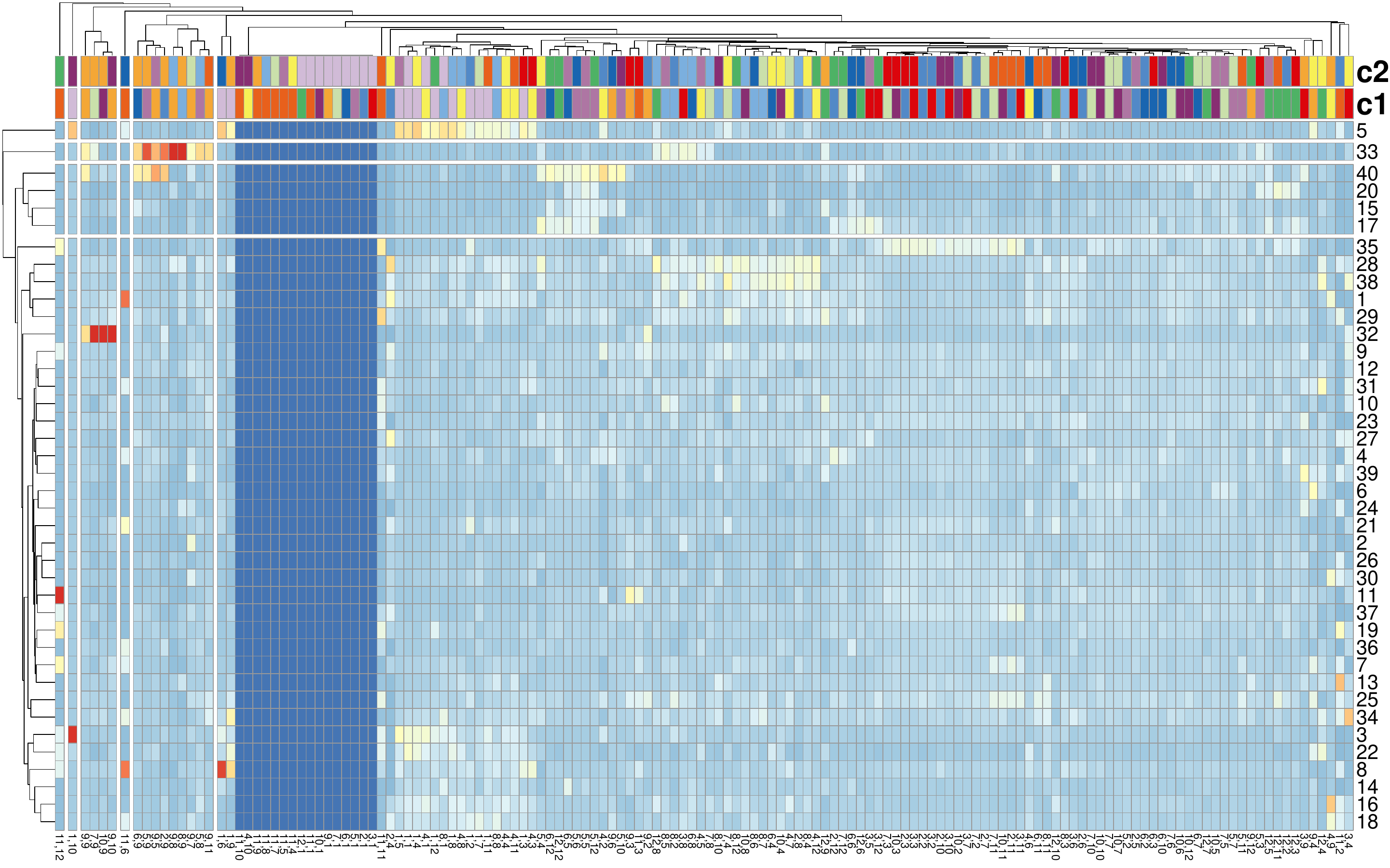}
    \end{tabular}
    
    \caption{A heatmap of the probabilities of the 40 different topics (on the y axis) for the different textual exchanges between clusters (on the x axis; the first coordinate is the cluster that quotes/reposts, the second the cluster being quoted/reposted);}
    \label{fig: lda_analysis}
\end{figure}
\clearpage

A topic distribution was estimated for each pair of cluster interactions $(g_1,g_2)_{1\leq g_1\leq G,1\leq g_2\leq G}$. This was done by selecting the corpus of a pair, and by summing the rows of $\gamma$ that correspond to the words of this corpus, and then by normalising these rows.  The results are shown in Figure \ref{fig: lda_analysis}. While there are few topics that specifically align to one cluster, it does seem that cluster 1 quotes or reposts content that is often assigned to topic 5. This is particularly interesting because cluster 1 is the closest to the sociological definition of a community, with a fairly high internal connection probability. This may well be explained by topic 5, since this topic is associated with a specific movement, the $\#\text{SaveSoil}$ movement. This movement is also detected when using a different model choice criterion, the ICL, though it is detected in a much more fragmented way.

\subsection*{\textCfour \label{ssec: Cfour}}

The integrated complete data log likelihood \citep[ICL,][]{biernacki_assessing_2002, daudin_mixture_2008, greedpackage} is defined as the logarithm of the joint posterior density of the data and the allocation vector matrix \begin{align*}
    \text{ICL}(C,G)=\log(p(Y,C|G))=\log(p(Y|G,C))+\log(p(C|G)).
\end{align*} \citet{greedpackage} showed how to maximise the ICL with respect to $C$ and $G$ via a greedy inference algorithm and implemented this algorithm in the R package greed. We chose $G=20$ as the initial value for the number of clusters and ran this algorithm. The resulting estimate gives 37 different clusters.

\begin{table}[]
    \centering
    \begin{tabular}{ccccccccccccccc}
        cluster&\textbf{1}&\textbf{2}&\textbf{3}&4&5&6&7&8&\textbf{9}&10&11&12&13\\\hline
        size&\textbf{3}&\textbf{1}&\textbf{5}&22&41&79&20&50&\textbf{9}&48&236&117&40\\\\
        cluster&14&15&16&17&18&19&20&21&22&23&24&25&26\\\hline
        size&12&11&42&123&220&54&13&81&468&30&32&60&96\\\\
        cluster&27&28&29&30&31&32&33&34&35&36&37&\\\hline
        size&314&773&137&121&231&928&894&229&186&1027&4226&&
    \end{tabular}
    
    \caption{estimated cluster sizes of the COP28 dataset using the greed algorithm; Clusters with 10 or less users assigned to them are emphasised.}
    \label{tab: cop28_estimated_cluster_sizes_greed}
\end{table} 

    \begin{table}[]
    \centering
    \begin{tabular}{ccccccccccccc}
        &&&&&&{\huge $\hat{\tau}$}&&&&&\\\hline
        $\hat{\tau}_1$ & $\hat{\tau}_2$ & $\hat{\tau}_3$ & $\hat{\tau}_4$ & $\hat{\tau}_5$ & $\hat{\tau}_6$ & $\hat{\tau}_7$ & $\hat{\tau}_8$ & $\hat{\tau}_9$ & $\hat{\tau}_{10}$ & $\hat{\tau}_{11}$ & $\hat{\tau}_{12}$ & $\hat{\tau}_{13}$\\\hline
         \textbf{0.03}&\textbf{0.01}&\textbf{0.05}&0.20&0.37&0.72&0.18&0.46&\textbf{0.08}&0.44&2.15&1.07&0.36\\
        $\hat{\tau}_{14}$ & $\hat{\tau}_{15}$ & $\hat{\tau}_{16}$ & $\hat{\tau}_{17}$ & \textbf{$\hat{\tau}_{18}$} & \textbf{$\hat{\tau}_{19}$} & \textbf{$\hat{\tau}_{20}$} & $\hat{\tau}_{21}$ & $\hat{\tau}_{22}$ & $\hat{\tau}_{23}$ & $\hat{\tau}_{24}$ & $\hat{\tau}_{25}$ & \textbf{$\hat{\tau}_{26}$}\\\hline
        0.11&0.10&0.38&1.12&2.00&0.49&0.12&0.74&4.26&0.27&0.29 &0.55&0.87\\
        $\hat{\tau}_{27}$ & \textbf{$\hat{\tau}_{28}$} & \textbf{$\hat{\tau}_{29}$} & $\hat{\tau}_{30}$ & $\hat{\tau}_{31}$ & $\hat{\tau}_{32}$ & $\hat{\tau}_{33}$ & $\hat{\tau}_{34}$ & $\hat{\tau}_{35}$ & $\hat{\tau}_{36}$ & $\hat{\tau}_{37}$ &&\\\hline
        2.86&7.04&1.25&1.10&2.10&8.45&8.14&2.09&1.69&9.35&38.49&&
    \end{tabular}
    
    \caption{estimates (in percent) of the SBM parameters using the greed algorithm; In the context of the COP28 dataset, $\hat{\tau}_g$ estimates the probability that a node chosen at random belongs to cluster $g$. Estimates are rounded to the second decimal place. Clusters with 10 or less users assigned to them are emphasised.}
    \label{tab: cop28_sbm_parameter_estimates_greed01}
\end{table}

    \begin{table}[]
    \centering
    
    \begin{tabular}{ccccccccccccc}
    &&&&&&{\huge$\hat{\mu}$}&&&&&\\\hline
    $\hat{\mu}_1$ & $\hat{\mu}_2$ & $\hat{\mu}_3$ & $\hat{\mu}_4$ & $\hat{\mu}_5$ & $\hat{\mu}_6$ & $\hat{\mu}_7$ & $\hat{\mu}_8$ & $\hat{\mu}_9$ & $\hat{\mu}_{10}$ & $\hat{\mu}_{11}$ & $\hat{\mu}_{12}$ & $\hat{\mu}_{13}$\\\hline
\textbf{22.22} & \textbf{0} & \textbf{0} & 0 & 0 & 0.42 & 3.33 & 0.67 & \textbf{0} & 0 & 0 & 0 & 0\\
\textbf{100} & \textbf{0} & \textbf{80} & 45.45 & 12.2 & 7.59 & 35 & 14 & \textbf{33.33} & 2.08 & 1.69 & 2.56 & 5\\
\textbf{26.67} & \textbf{20} & \textbf{24} & 36.36 & 15.61 & 11.9 & 10 & 2.4 & \textbf{0} & 0.42 & 1.1 & 1.71 & 1.5\\
\textbf{31.82} & \textbf{9.09} & \textbf{37.27} & 18.6 & 8.2 & 7.13 & 6.82 & 1.64 & \textbf{1.52} & 0.38 & 0.46 & 1.63 & 0.91\\
\textbf{53.66} & \textbf{14.63} & \textbf{60.49} & 33.92 & 15.41 & 11.82 & 11.95 & 2.49 & \textbf{1.9} & 0.05 & 0.9 & 2.88 & 2.13\\
\textbf{14.35} & \textbf{3.8} & \textbf{23.54} & 9.49 & 2.96 & 2.69 & 2.78 & 0.63 & \textbf{0} & 0.05 & 0.1 & 0.38 & 0.22\\
\textbf{25} & \textbf{0} & \textbf{1} & 0.68 & 0.37 & 0.13 & 5.25 & 1.3 & \textbf{2.22} & 0 & 0 & 0 & 0\\
\textbf{24.67} & \textbf{4} & \textbf{2.4} & 1.18 & 0.39 & 0.61 & 2.4 & 1.92 & \textbf{2.89} & 0.08 & 0.09 & 0.12 & 0\\
\textbf{33.33} & \textbf{22.22} & \textbf{0} & 0 & 0 & 0 & 2.78 & 3.56 & \textbf{4.94} & 0.23 & 0.14 & 0.19 & 0\\
\textbf{73.61} & \textbf{18.75} & \textbf{30.83} & 14.2 & 3.56 & 4.11 & 17.6 & 5.42 & \textbf{8.8} & 0.48 & 0.34 & 0.61 & 0.36\\
\textbf{32.2} & \textbf{7.2} & \textbf{34.49} & 15.16 & 5.61 & 4.54 & 6.74 & 1.08 & \textbf{1.18} & 0.05 & 0.29 & 0.95 & 0.85\\
\textbf{45.3} & \textbf{12.82} & \textbf{52.31} & 29.68 & 12.84 & 9.5 & 11.2 & 2.07 & \textbf{1.9} & 0.07 & 0.54 & 2.02 & 1.67\\
\textbf{65.83} & \textbf{40} & \textbf{74} & 53.07 & 22.56 & 20.51 & 25.37 & 4.1 & \textbf{7.78} & 0.1 & 1.08 & 4.55 & 3.62\\
\textbf{86.11} & \textbf{66.67} & \textbf{70} & 51.52 & 22.97 & 21.2 & 38.75 & 21.67 & \textbf{25.93} & 2.78 & 2.19 & 5.27 & 5.21\\
\textbf{72.73} & \textbf{100} & \textbf{0} & 1.24 & 0 & 0.46 & 10.91 & 7.09 & \textbf{65.66} & 2.08 & 0.12 & 0 & 0\\
\textbf{47.62} & \textbf{61.9} & \textbf{0.48} & 0.11 & 0.06 & 0.09 & 2.86 & 1.86 & \textbf{25.93} & 0.4 & 0.11 & 0 & 0\\
\textbf{29.81} & \textbf{39.02} & \textbf{0} & 0.04 & 0 & 0.01 & 0.37 & 0.28 & \textbf{9.12} & 0.15 & 0 & 0 & 0\\
\textbf{2.27} & \textbf{0} & \textbf{0} & 0.02 & 0 & 0 & 0.05 & 0.05 & \textbf{0.15} & 0.01 & 0 & 0 & 0\\
\textbf{12.96} & \textbf{7.41} & \textbf{0} & 0 & 0 & 0 & 0.37 & 0.26 & \textbf{3.91} & 0.12 & 0.01 & 0 & 0\\
\textbf{5.13} & \textbf{0} & \textbf{1.54} & 0.35 & 0 & 0 & 0.38 & 0.46 & \textbf{1.71} & 0 & 0.03 & 0 & 0\\
\textbf{4.53} & \textbf{1.23} & \textbf{0.74} & 0.73 & 0.03 & 0.05 & 0.31 & 0.2 & \textbf{0.14} & 0.03 & 0 & 0 & 0\\
\textbf{1.78} & \textbf{0} & \textbf{0.34} & 0.14 & 0.02 & 0.02 & 0.33 & 0.08 & \textbf{0.05} & 0 & 0 & 0 & 0.03\\
\textbf{1.11} & \textbf{0} & \textbf{0.67} & 0 & 0 & 0.08 & 0.17 & 0.13 & \textbf{0} & 0.07 & 0 & 0 & 0\\
\textbf{0} & \textbf{0} & \textbf{3.75} & 0.14 & 0 & 0.08 & 0.16 & 0.06 & \textbf{0} & 0.2 & 0.01 & 0 & 0\\
\textbf{7.22} & \textbf{0} & \textbf{0} & 0.76 & 0.04 & 0.02 & 0.17 & 0.07 & \textbf{0.19} & 0 & 0 & 0 & 0\\
\textbf{5.21} & \textbf{0} & \textbf{0.62} & 0.05 & 0 & 0.01 & 0.31 & 0.21 & \textbf{0.23} & 0.04 & 0 & 0 & 0\\
\textbf{5.52} & \textbf{0.64} & \textbf{0.89} & 0.2 & 0 & 0.03 & 1.11 & 0.17 & \textbf{0.28} & 0.01 & 0 & 0.01 & 0\\
\textbf{4.96} & \textbf{0.39} & \textbf{1.35} & 0.33 & 0.06 & 0.08 & 0.9 & 0.17 & \textbf{0.13} & 0.01 & 0.01 & 0.01 & 0.04\\
\textbf{12.9} & \textbf{1.46} & \textbf{4.23} & 1.66 & 0.36 & 0.39 & 3.72 & 0.69 & \textbf{0.41} & 0.06 & 0.05 & 0.06 & 0.15\\
\textbf{23.69} & \textbf{0.83} & \textbf{9.09} & 3.61 & 0.75 & 0.95 & 2.52 & 1.22 & \textbf{0.73} & 0.12 & 0.05 & 0.09 & 0.14\\
\textbf{33.19} & \textbf{9.96} & \textbf{5.63} & 1.55 & 0.42 & 0.32 & 3.98 & 2.16 & \textbf{3.56} & 0.35 & 0.05 & 0.08 & 0.04\\
\textbf{10.09} & \textbf{1.94} & \textbf{0.99} & 0.19 & 0.07 & 0.05 & 0.77 & 0.76 & \textbf{1.4} & 0.04 & 0.01 & 0.01 & 0.01\\
\textbf{11.56} & \textbf{2.8} & \textbf{8.84} & 3.16 & 0.98 & 0.89 & 1.97 & 0.33 & \textbf{0.32} & 0.04 & 0.08 & 0.19 & 0.21\\
\textbf{15.28} & \textbf{5.24} & \textbf{17.99} & 7.42 & 2.61 & 2.26 & 1.92 & 0.64 & \textbf{0.53} & 0.06 & 0.16 & 0.42 & 0.19\\
\textbf{1.61} & \textbf{1.08} & \textbf{1.94} & 0.56 & 0.22 & 0.1 & 0.32 & 0.12 & \textbf{0} & 0 & 0.03 & 0.04 & 0.01\\
\textbf{0.26} & \textbf{0} & \textbf{0.41} & 0.08 & 0.01 & 0.04 & 0.06 & 0.02 & \textbf{0} & 0 & 0.01 & 0.01 & 0\\
\textbf{1.15} & \textbf{0.14} & \textbf{0.22} & 0.04 & 0.01 & 0.01 & 0.1 & 0.06 & \textbf{0.02} & 0 & 0 & 0 & 0\\
    \end{tabular}
    
    \caption{estimates of the columns 1 to 13 of the SBM parameter $\mu$ (in percent); In the context of the COP28 dataset, $\hat{\mu}_{g,k}$ estimates the probability that a user in cluster $g$ quotes or reposts a user in cluster $k$. Estimates are rounded to the second decimal place. Columns corresponding to clusters with 10 or less users assigned to them are emphasised.}
    \label{tab: cop28_sbm_parameter_estimates_greed02}
\end{table} 

    \begin{table}[]
    \centering
    
    \begin{tabular}{ccccccccccccc}
    &&&&&&{\huge$\hat{\mu}$}&&&&&\\\hline
    $\hat{\mu}_{14}$ & $\hat{\mu}_{15}$ & $\hat{\mu}_{16}$ & $\hat{\mu}_{17}$ & \textbf{$\hat{\mu}_{18}$} & \textbf{$\hat{\mu}_{19}$} & \textbf{$\hat{\mu}_{20}$} & $\hat{\mu}_{21}$ & $\hat{\mu}_{22}$ & $\hat{\mu}_{23}$ & $\hat{\mu}_{24}$ & $\hat{\mu}_{25}$ & \textbf{$\hat{\mu}_{26}$}\\\hline
    0 & 0 & 0 & 0 & 0 & 0 & 0 & 0 & 0 & 0 & 0 & 0 & 0\\
0 & 0 & 0 & 0 & 0 & 0 & 0 & 0 & 0.21 & 0 & 0 & 0 & 0\\
5 & 0 & 0 & 0 & 0 & 0 & 0 & 0.49 & 0.04 & 0 & 0.62 & 0 & 0\\
0.38 & 0 & 0 & 0 & 0 & 0 & 0 & 0.56 & 0.06 & 0 & 0 & 0 & 0.05\\
1.02 & 0 & 0 & 0 & 0 & 0.05 & 0.38 & 0.45 & 0.05 & 0 & 0 & 0 & 0\\
0.32 & 0 & 0 & 0 & 0 & 0.02 & 0.1 & 0.08 & 0.01 & 0 & 0.04 & 0 & 0\\
0 & 0 & 0 & 0 & 0 & 0 & 1.15 & 0.06 & 0.02 & 0 & 0 & 0 & 0\\
0.33 & 0 & 0 & 0 & 0.02 & 0.04 & 1.38 & 0.25 & 0.07 & 0 & 0 & 0 & 0.08\\
0 & 0 & 0 & 0 & 0.05 & 0.82 & 3.42 & 0 & 0 & 0 & 0 & 0 & 0\\
1.39 & 0 & 0.05 & 0 & 0.04 & 0.69 & 3.53 & 1.65 & 0.3 & 0.42 & 0.07 & 0 & 0.26\\
0.32 & 0 & 0 & 0 & 0 & 0.02 & 0.23 & 0.26 & 0.02 & 0.01 & 0.03 & 0 & 0.01\\
0.36 & 0 & 0 & 0 & 0.01 & 0.02 & 0.26 & 0.37 & 0.03 & 0 & 0 & 0 & 0.01\\
1.04 & 0 & 0 & 0 & 0.03 & 0.05 & 0.96 & 0.74 & 0.05 & 0.08 & 0 & 0.21 & 0.05\\
6.94 & 0 & 0 & 0 & 0.11 & 0.93 & 10.26 & 2.06 & 0.45 & 7.5 & 6.51 & 0 & 0.69\\
0 & 56.2 & 37.01 & 23.28 & 33.14 & 79.12 & 94.41 & 0.45 & 0.02 & 0 & 0 & 0 & 0.09\\
0 & 26.41 & 12.47 & 5.48 & 8.52 & 38.05 & 78.02 & 0.47 & 0.01 & 0 & 0 & 0 & 0\\
0 & 4.8 & 1.2 & 0.85 & 1.08 & 9.08 & 56.47 & 0.12 & 0 & 0 & 0 & 0 & 0\\
0 & 0.21 & 0.04 & 0.02 & 0.02 & 0.29 & 3.92 & 0.02 & 0 & 0 & 0 & 0 & 0\\
0 & 0.67 & 0.4 & 0.35 & 0.44 & 2.78 & 24.36 & 0.02 & 0 & 0 & 0 & 0 & 0\\
0 & 0 & 0 & 0 & 0.03 & 0.57 & 7.1 & 0.38 & 0.03 & 0 & 0 & 0 & 0\\
0 & 0 & 0 & 0 & 0 & 0 & 0 & 0.91 & 0.2 & 0 & 0 & 0 & 0.01\\
0 & 0 & 0 & 0 & 0 & 0 & 0.02 & 0.73 & 0.12 & 0 & 0 & 0 & 0.01\\
0 & 0 & 0 & 0 & 0 & 0 & 0 & 0 & 0 & 0 & 0 & 0 & 0\\
2.86 & 0 & 0 & 0 & 0 & 0 & 0 & 0 & 0.01 & 45.42 & 46.09 & 0 & 0\\
0 & 0 & 0 & 0 & 0 & 0 & 0 & 0.04 & 0 & 0 & 0 & 14.03 & 0\\
0 & 0 & 0 & 0 & 0 & 0 & 0.08 & 0.66 & 0.03 & 0 & 0 & 0.02 & 5.34\\
0.05 & 0 & 0 & 0 & 0 & 0.05 & 0.02 & 0.2 & 0.01 & 0 & 0 & 0 & 0\\
0.01 & 0 & 0 & 0 & 0 & 0.01 & 0.21 & 2.42 & 0.47 & 0 & 0 & 0 & 0.02\\
0 & 0 & 0 & 0 & 0 & 0.03 & 0.39 & 5.5 & 2.17 & 0 & 0 & 0.01 & 0.14\\
0 & 0 & 0 & 0 & 0.01 & 0.05 & 0.7 & 1.92 & 0.27 & 0.03 & 0 & 0.01 & 0.02\\
0.4 & 0 & 0 & 0 & 0.01 & 0.11 & 1.07 & 0.15 & 0.02 & 0.04 & 0.03 & 0.01 & 0.17\\
0.01 & 0 & 0 & 0 & 0 & 0.04 & 0.34 & 0.04 & 0.01 & 0 & 0 & 0.02 & 0.01\\
0.06 & 0 & 0 & 0 & 0 & 0.01 & 0.06 & 0.1 & 0.02 & 0 & 0 & 0 & 0\\
0.25 & 0 & 0 & 0 & 0 & 0 & 0 & 0.04 & 0 & 0.01 & 0 & 0 & 0\\
0.04 & 0 & 0 & 0 & 0 & 0 & 0.04 & 0.06 & 0 & 0 & 0 & 0 & 0\\
0 & 0 & 0 & 0 & 0 & 0 & 0.01 & 0 & 0 & 0 & 0 & 0 & 0\\
0.01 & 0 & 0 & 0 & 0 & 0 & 0.04 & 0.2 & 0.02 & 0 & 0 & 0 & 0\\
    \end{tabular}
    
    \caption{estimates of the columns 14 to 26 of the SBM parameter $\mu$ (in percent); In the context of the COP28 dataset, $\hat{\mu}_{g,k}$ estimates the probability that a user in cluster $g$ quotes or reposts a user in cluster $k$. Estimates are rounded to the second decimal place. Columns corresponding to clusters with 10 or less users assigned to them are emphasised (there aren't any in this case).}
    \label{tab: cop28_sbm_parameter_estimates_greed03}
\end{table}
\begin{table}
    \begin{tabular}{ccccccccccc}
    &&&&&&{\huge$\hat{\mu}$}&&&\\\hline
    $\hat{\mu}_{27}$ & $\hat{\mu}_{28}$ & $\hat{\mu}_{29}$ & $\hat{\mu}_{30}$ & $\hat{\mu}_{31}$ & $\hat{\mu}_{32}$ & $\hat{\mu}_{33}$ & $\hat{\mu}_{34}$ & $\hat{\mu}_{35}$ & $\hat{\mu}_{36}$ & $\hat{\mu}_{37}$ \\\hline
0 & 0 & 0 & 0 & 0.14 & 0 & 0 & 0 & 0 & 0 & 0\\
0.64 & 0 & 0 & 0 & 3.9 & 0.54 & 0.67 & 1.31 & 2.69 & 0.19 & 0.09\\
0 & 0 & 0 & 0 & 0.43 & 0 & 0.27 & 3.06 & 1.4 & 0.12 & 0\\
0 & 0 & 0 & 0.23 & 0.31 & 0.03 & 0.25 & 1.85 & 1.59 & 0.34 & 0.02\\
0.01 & 0 & 0 & 0.1 & 0.37 & 0.08 & 0.37 & 3.32 & 2.36 & 0.47 & 0\\
0.01 & 0 & 0.01 & 0.09 & 0.09 & 0.02 & 0.06 & 0.76 & 0.55 & 0.07 & 0\\
0.02 & 0 & 0 & 0 & 0.22 & 0.07 & 0.02 & 0.11 & 0.27 & 0.04 & 0\\
0.06 & 0.01 & 0 & 0.08 & 0.57 & 0.38 & 0.02 & 0.19 & 0.29 & 0.05 & 0.01\\
0 & 0 & 0 & 0.09 & 0.67 & 0.44 & 0 & 0 & 0.12 & 0 & 0\\
0.25 & 0.02 & 0.17 & 0.88 & 2.08 & 0.67 & 0.26 & 1.15 & 1.67 & 0.5 & 0.08\\
0.02 & 0 & 0.02 & 0.04 & 0.19 & 0.05 & 0.18 & 0.97 & 1.2 & 0.22 & 0.01\\
0 & 0 & 0 & 0.06 & 0.32 & 0.06 & 0.26 & 2.36 & 2.08 & 0.38 & 0\\
0.02 & 0 & 0 & 0.21 & 0.69 & 0.1 & 0.71 & 5.08 & 5.65 & 0.92 & 0.04\\
0.16 & 0.04 & 0 & 0.83 & 6.67 & 2.03 & 1.1 & 7.35 & 7.48 & 2.26 & 0.13\\
0 & 0.02 & 0.07 & 0 & 1.65 & 0.22 & 0.02 & 0.12 & 0.24 & 0.92 & 0.01\\
0.01 & 0.01 & 0 & 0 & 0.14 & 0.01 & 0 & 0 & 0.06 & 0.13 & 0.01\\
0 & 0 & 0 & 0.01 & 0.03 & 0 & 0 & 0 & 0 & 0.01 & 0\\
0 & 0 & 0 & 0 & 0.01 & 0 & 0 & 0 & 0 & 0 & 0\\
0 & 0 & 0 & 0 & 0.06 & 0.01 & 0 & 0.02 & 0.01 & 0.01 & 0\\
0 & 0.01 & 0 & 0.06 & 0.17 & 0.03 & 0 & 0 & 0 & 0.02 & 0.01\\
0 & 0.01 & 0.01 & 0.03 & 0.02 & 0 & 0.01 & 0.03 & 0.08 & 0 & 0.01\\
0 & 0.01 & 0.08 & 0.05 & 0 & 0 & 0 & 0 & 0.02 & 0 & 0.01\\
0 & 0 & 0 & 0 & 0.04 & 0 & 0 & 0.03 & 0.04 & 0 & 0\\
0 & 0 & 0 & 0 & 0.03 & 0.03 & 0 & 0.1 & 0.02 & 0.01 & 0.01\\
0 & 0 & 0 & 0 & 0.02 & 0.03 & 0 & 0 & 0.04 & 0 & 0.01\\
0 & 0.01 & 0.06 & 0 & 0.11 & 0.03 & 0 & 0 & 0 & 0.01 & 0.01\\
0.98 & 0 & 0 & 0.03 & 0.03 & 0.02 & 0 & 0.01 & 0.03 & 0 & 0.03\\
0.01 & 0.07 & 0.23 & 0.18 & 0.01 & 0.01 & 0 & 0.02 & 0.09 & 0.01 & 0.04\\
0 & 0.37 & 1.79 & 0.2 & 0.07 & 0.01 & 0.01 & 0.07 & 0.29 & 0.02 & 0.14\\
0.07 & 0.03 & 0.05 & 2.1 & 0.16 & 0.1 & 0.03 & 0.16 & 0.23 & 0.04 & 0.07\\
0.03 & 0 & 0 & 0.06 & 1.04 & 0.4 & 0.03 & 0.12 & 0.2 & 0.08 & 0.01\\
0.01 & 0 & 0 & 0.03 & 0.3 & 0.17 & 0 & 0.02 & 0.04 & 0.01 & 0.01\\
0.01 & 0 & 0 & 0.02 & 0.07 & 0.01 & 0.09 & 0.29 & 0.35 & 0.09 & 0.01\\
0 & 0 & 0 & 0.04 & 0.11 & 0.02 & 0.13 & 0.83 & 0.49 & 0.13 & 0\\
0 & 0 & 0 & 0.01 & 0.01 & 0 & 0.02 & 0.05 & 0.1 & 0.02 & 0\\
0 & 0 & 0 & 0 & 0.01 & 0 & 0.01 & 0.02 & 0.03 & 0.01 & 0\\
0.01 & 0 & 0 & 0.02 & 0.01 & 0 & 0 & 0 & 0.02 & 0 & 0.01\\
    \end{tabular}
    
    \caption{estimates of the columns 27 to 37 of the SBM parameter $\mu$ (in percent); In the context of the COP28 dataset, $\hat{\mu}_{g,k}$ estimates the probability that a user in cluster $g$ quotes or reposts a user in cluster $k$. Estimates are rounded to the second decimal place. Columns corresponding to clusters with 10 or less users assigned to them are emphasised (there aren't any in this case).}
    \label{tab: cop28_sbm_parameter_estimates_greed04}
\end{table}

\begin{figure}
    \centering
    \includegraphics[scale=.5]{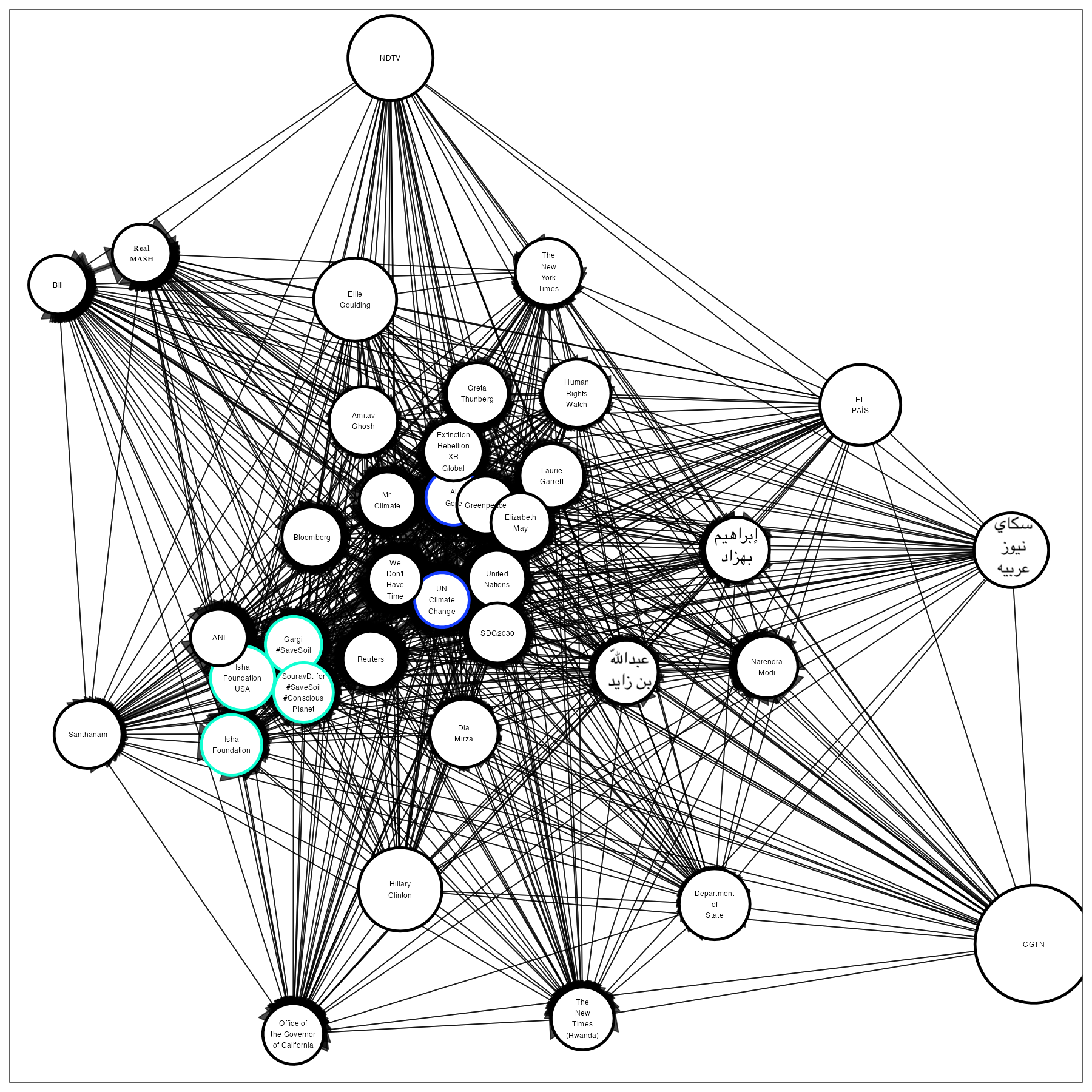}
    \caption{the metagraph of the COP28 dataset for $G=37$ clusters, the number obtained by maximising the ICL via the greed algorithm; Clusters are represented by nodes. Large edges indicate large connection probabilities between clusters, while large nodes indicate a large relative size of the cluster; Among other clusters, $"$Al Gore$"$ and $"$UN Climate Change$"$ represent the core of the network, while $"$Gargi \# SaveSoil$"$, $"$SouravD. for \# SaveSoil \# Conscious Planet$"$, $"$Isha Foundation$"$ and $"$Isha Foundation USA$"$ are associated with the $\#\text{SaveSoil}$ movement.}
    \label{fig: greed_metagraph}
\end{figure}

The sizes of the 37 different clusters obtained by the package greed for the COP28 dataset are shown in Table \ref{tab: cop28_estimated_cluster_sizes_greed}. Notably, there are a variety of clusters with a very small size. Indeed, one cluster even has a size as small as one. These smaller clusters all seem to have the property that it is quite likely to connect to them even though the users assigned to them make up a small portion of the dataset. This can be seen by observing the estimates of the SBM parameters shown in Figures \ref{tab: cop28_sbm_parameter_estimates_greed01}, \ref{tab: cop28_sbm_parameter_estimates_greed02}, \ref{tab: cop28_sbm_parameter_estimates_greed03} and \ref{tab: cop28_sbm_parameter_estimates_greed04}, which were obtained by maximising the complete-data loglikelihood with respect to these parameters, given the clustering obtained by greed, just like in the previous section. The resulting metagraph is shown in Figure \ref{fig: greed_metagraph}. It seems like the $"$core$"$ (users which have a high probability of being quoted or reposted) is much more fractured, with several small clusters having this property. For example, the clusters containing the users $"$Al Gore$"$ and $"$UN Climate Change$"$ are split, while these users belonged to cluster 11 for the clustering obtained by the THAMES, which completely represented the core in this case. In addition, there are 4 clusters associated with the $\#\text{SaveSoil}$ movement (two of which have users that directly use the hashtag, one of which is its central figure, Sadhguru, and one of which is its central institution in the US, the Isha Foundation USA). This indicates a further fracturing of the clusters found when using the THAMES. In combination, we take these two differences to the clustering found by the THAMES - the fracturing of the core and the fracturing of the $\#\text{SaveSoil}$ movement into multiple clusters - as an indication that the clustering found by the THAMES is more robust than the one found by the ICL. This is because the former includes one movement, the $\#\text{SaveSoil}$ movement, and one core, which includes 5 members, while the later fractured the movement and the core into multiple clusters.

\subsection*{\textCfive \label{ssec: Cfive}}

As shown in the main document, the clusters found by using the THAMES and Gibbs sampling on the COP28 dataset essentially all follow a star pattern: while it is unlikely that users in the same cluster quote/repost each other, they do all tend to quote/repost cluster 11, which only contains the users UN Climate Change, Al Gore, COP28 UAE, Loss and Damage Collaboration (L\&DC) and António Guterres. Thus, cluster 11 appears to be the manifestation of the dominant voices of the COP28, which are quoted and reposted by all other users. This indicates a core-periphery structure, which creates a dichotomy between clusters that include bots, grassroots movements, and influencers on one side, and the COP28 institution on the other side, revealing a particularly extreme core-periphery structure.

The core-periphery structure present in the COP28 dataset is by and large confirmed by our textual analysis, as it seems like almost all topics are fairly similar and centred around climate change, with little distinction between topics that are reposted/cited by users from the same cluster. While our textual analysis does indicate that there may be one cluster of the COP28 dataset that is somewhat close to the standard definition of a community, cluster 1, even this cluster quotes/reposts another cluster, cluster 8, more often than it reposts itself. This is likely because the X account of the central figure of the \#SaveSoil movement, Sadhguru, is assigned to cluster 8. Thus, even cluster 1 exhibits aspects of the core/periphery distinction central to this dataset, since its users repost/cite users of another cluster more often than users assigned to the same cluster. It should also be considered that Save Soil is a campaign run by the Isha Foundation, with its own pavilion at COP28 and coordinated messaging. Its founder said publicly during the summit that “Our messages are being tweeted daily by a few thousand people, this needs to become millions”, as reported by {\color{blue}\href{https://www.thenationalnews.com/climate/cop28/2023/12/10/protect-our-soil-to-save-the-planet-says-indian-climate-guru-sadhguru-at-cop28}{The National}}. So this is a deliberate occupation of the space, not a spontaneous one. It should also be noted that the model separates the central figure (cluster 8) from those who amplify him (cluster 1), which is, arguably, the structurally correct answer, while the clustering that maximises the ICL splits the same movement across four clusters. Thus, these results could be seen as an external validation of the clustering found by the THAMES, compared to the clustering found by maximising the ICL.

\subsection*{\textCsix \label{ssec: Csix}}

\begin{figure}
    \centering
    \includegraphics[scale=.5]{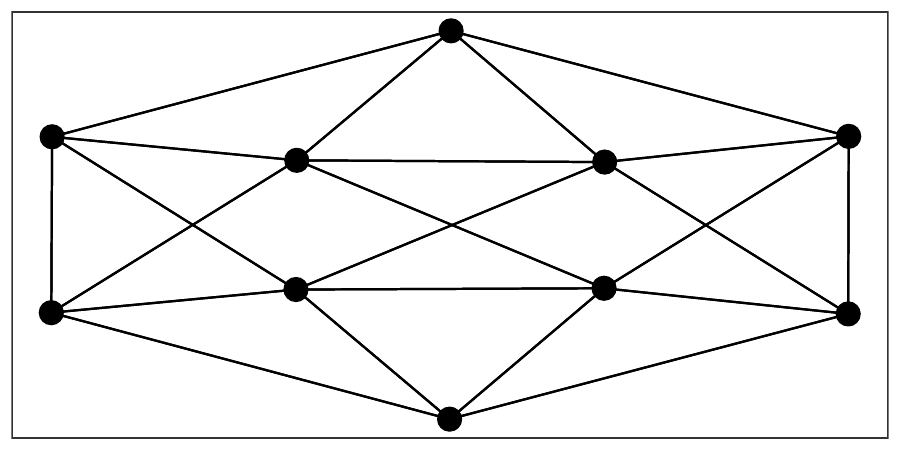}
    \caption{The specific dataset used in Section 4.3}
    \label{fig: specific dataset}
\end{figure}

The specific dataset used in Section 4.3 can be represented by the graph shown in Figure \ref{fig: specific dataset}. The posterior probability of the MAP (and its label-switched version) was around 15 percent for this dataset when fitting the SBM with $G=2$. This probability is much smaller than the usual probabilities obtained when simulating datasets of such a small size, which were often larger than 50 percent, in our experience.

\section*{\suppDtext \label{sec: suppD}}

The latent block model (LBM) has a very similar structure to the SBM, with the difference that the size of the rows $n$ may be different from the size of the columns $m$, and that there are two allocation vectors $C$ and $C'$ of sizes $G$ and $K$, which indicate the row and column clusters, respectively:\begin{align*}&C_i|\tau\stackrel{\text{i.i.d}}\sim\text{Multinom}_G(1,\tau),\quad\sum_{g=1}^GC_{i,g}=1,\quad\forall i=1,\dots,n,\\
&C_j'|\tau'\stackrel{\text{i.i.d}}\sim\text{Multinom}_{G'}(1,\tau'),\quad\sum_{k=1}^KC_{j,g}'=1,\quad\forall j=1,\dots,m,\\
    &Y_{i,j}|\mu,C_{i,g}C_{j,k}'=1\sim f(\mu_{g,k}),\quad \forall(i,j),i\neq j.
\end{align*} $f$ can denote any density, and so the LBM can be defined for many types of data, such as count data \citep{govaert_latent_2010}, categorical data \citep{keribin_estimation_2015} or continuous data \citep{lomet_selectiondemodele_2012}. For many versions of the LBM, Gibbs sampling is possible and the collapsed distributions $p(C,C'|G,K)$ and $p(Y|G,K,C,C')$ can also be derived; see, e.g., \citet{WyFr12-collapsed_infinite_LBM,lomet_selectiondemodele_2012,keribin_estimation_2015}. 

Thus, the THAMES can be defined completely analogously, as {\small\begin{align*}
    \hat{Z}_{\text{LBM}}^{-1}=\frac{1}{G!}\sum^{G!}_{o=1}\frac{1}{K!}\sum^{K!}_{s=1}\frac{1}{T/2}\sum^{T}_{\substack{t=T/2+1\\(P_o(C^{(t)}),P_s'(C^{(t)}{}'))\in B_{\hat{C},\hat{r},\hat{C}',\hat{r}',\alpha}}}\frac{1/|B_{\hat{C},\hat{r},\hat{C}',\hat{r}',\alpha}|}{p(Y|G,K,C^{(t)},C^{(t)}{}')p(C^{(t)},C^{(t)}{}'|G,K)},
\end{align*}} with $(\hat{r},\hat{r}')$ being set such that they minimise {\small\begin{align*}
    \left(\frac{\prod^n_{i=1}\sum_{C_i\in \{\hat{C}^{(1)}_i,\dots,\hat{C}^{(\hat{r})}_i\}}\frac{1}{\text{Multinom}_G(C_i;1,\hat{z}_i)}}{\left(\prod^n_{i=1}|\{\hat{C}_i^{(1)},\dots,\hat{C}_i^{(\hat{r})}\}|\right)^2}\right)\cdot\left(\frac{\prod^m_{j=1}\sum_{C_j'\in \{\hat{C}^{(1)}_j{}',\dots,\hat{C}^{(\hat{r}')}_j{}'\}}\frac{1}{\text{Multinom}_G(C_j';1,\hat{z}_j')}}{\left(\prod^m_{j=1}|\{\hat{C}_j^{(1)}{}',\dots,\hat{C}_j^{(\hat{r}')}{}'\}|\right)^2}\right),
\end{align*}} and $(\hat{C},\hat{C}')$ being constructed by maximising the unnormalised posterior density on the first half of the MCMC sample, while\begin{align*}
    B_{\hat{C},\hat{r},\hat{C}',\hat{r}',\alpha}=E_{\hat{C},\hat{r},\hat{C}',\hat{r}'}\cap H_\alpha,
\end{align*}
where $H_\alpha$ denotes a $\alpha$-HPD-region and
$E_{\hat{C},\hat{r},\hat{C}',\hat{r}'}$ is the corresponding Cartesian product, \begin{align*}
    E_{\hat{C},\hat{r},\hat{C}',\hat{r}'}&=(\{\hat{C}^{(1)}_1,\dots,\hat{C}^{(\hat{r})}_{1}\}\times\{\hat{C}^{(1)}_1{}',\dots,\hat{C}^{(\hat{r}')}_{1}{}'\})\times\dots\\&\times(\{\hat{C}^{(1)}_n,\dots,\hat{C}^{(\hat{r})}_{n}\}\times\{\hat{C}^{(1)}_n{}',\dots,\hat{C}^{(\hat{r}')}_{n}{}'\}).
\end{align*} Analogously to the THAMES for the SBM, this estimator can be quickly computed by choosing an ordering of the columns of $C$ and $C'$. However, the case of empty clusters cannot be easily sidestepped for all choices of the density $f$. In these cases, it is essential that there are no empty clusters in the chosen rows $\upsilon$ of the posterior sample, since the computational trick cannot be applied otherwise.


\end{document}